\documentclass[pdflatex,sn-nature]{sn-jnl}

\usepackage{physics}
\usepackage{graphicx}%
\usepackage{multirow}%
\usepackage{amsmath,amssymb,amsfonts}%
\usepackage{amsthm}%
\usepackage{mathrsfs}%
\usepackage[title]{appendix}%
\usepackage{acronym}
\usepackage{xcolor}%
\usepackage{textcomp}%
\usepackage{manyfoot}%
\usepackage{booktabs}%
\usepackage{algorithm}%
\usepackage{algorithmicx}%
\usepackage{algpseudocode}%
\usepackage{listings}%
\usepackage{stfloats}%
\usepackage{bbm}
\usepackage{bbold}
\usepackage[caption=false,font=normalsize,labelfont=sf,textfont=sf]{subfig}%
\usepackage{mathrsfs}
\usepackage{tabularx}

\theoremstyle{thmstyleone}%
\newtheorem{theorem}{Theorem}
\newtheorem{proposition}{Proposition}
\newtheorem{lemma}{Lemma}
\newtheorem{corollary}{Corollary}
\theoremstyle{thmstyletwo}%
\theoremstyle{thmstylethree}%
\newtheorem{definition}{Definition}%

\newacro{ILP}{integer linear program}
\newacro{QNT}{quantum network tomography}
\newacro{QFI}{Quantum Fisher information}
\newacro{QFIM}{quantum Fisher information matrix}
\newacro{QCRB}{quantum Cramér-Rao bound}
\newacro{CSQP}{cyclic sequential QNT protocol}
\newacro{MLE}{maximum likelihood estimator}
\newacro{MSE}{mean squared error}
\newacro{POVMs}{Positive Operator-Valued Measures}

\begin{document}

\title{Learnability, Identifiability, and Monitor Placement in Quantum Network Tomography}

\author*[1]{\fnm{Athira} \sur{Kalavampara Raghunadhan}}\email{kalavama@tcd.ie}\equalcont{These authors contributed equally to this work.}
\author[2]{\fnm{Matheus} \sur{Guedes De Andrade}}\email{mguedesdeand@umass.edu}\equalcont{These authors contributed equally to this work.}
\author[2]{\fnm{Don} \sur{Towsley}}\email{towsley@cs.umass.edu}
\author[3]{\fnm{Indrakshi} \sur{Dey}}\email{Indrakshi.Dey@waltoninstitute.ie}
\author[1]{\fnm{Daniel} \sur{ Kilper}}\email{dan.kilper@tcd.ie}
\author[1]{\fnm{Nicola} \sur{Marchetti}}\email{nicola.marchetti@tcd.ie}

\affil*[1]{\orgdiv{CONNECT Research Centre, School of Engineering}, \orgname{Trinity College Dublin}, \country{Ireland}}

\affil[2]{\orgdiv{Manning College of Information and Computer Science}, \orgname{University of Massachusetts}, \city{Amherst},  \country{USA}}

\affil[3]{\orgdiv{Walton Institute}, \orgname{South East Technological University}, \city{Waterford},  \country{Ireland}}

\abstract{Reliable quantum communication requires accurate characterization of the quantum links. This paper studies Quantum Network Tomography (QNT) under limited monitoring resources, where unknown link parameters are inferred from path-based measurements performed at monitor nodes. We introduce a cyclic sequential QNT protocol (CSQP) for arbitrary network topologies and develop a fixed-point learnability framework with an explicit algorithm for estimating link-level Werner parameters. We characterize identifiability through the rank of the path-link incidence matrix and show that the CSQP learnability conditions guarantee full rank and a nonsingular Quantum Fisher Information Matrix (QFIM). Building on this framework, we formulate monitor placement and measurement assignment as an optimization problem whose constraints enforce learnability and identifiability without topology-specific reformulation. Two Integer Linear Programming (ILP) formulations are introduced: Unconstrained QFIM-based formulation (QF), which maximizes QFIM-trace, and monitoring-overhead constrained QFIM formulation (QMF), which maximizes QFIM-trace subject to a per-monitor overhead constraint. Both formulations are evaluated on star and tree networks to compare monitor placements and measurement assignments. The results show that QMF distributes monitoring load evenly across all monitors and provides greater potential for parallel monitoring under resource constraints, while QF is more suitable when estimation information is prioritized, particularly in practical networks with non-uniform link noise.
}

\maketitle
\section*{Introduction}\label{sec1}
Quantum communication networks aim to distribute quantum states and entanglement between distant nodes through interconnected quantum links, enabling applications such as quantum key distribution \cite{bennett2014quantum}, distributed quantum computing \cite{buhrman2003distributed}, and quantum sensing \cite{zhang2021distributed}. Accurate characterization of quantum links is thus essential for evaluating link reliability, selecting appropriate communication paths, optimizing resource allocation, and ensuring that transmitted quantum states maintain the required fidelity.

However, quantum link characterization is challenging because direct access to every internal link or component of a network may be impractical, especially as networks increase in size and complexity. In classical networks, network tomography provides scalable and resource-efficient techniques for characterizing internal links using measurements performed at network monitors, without requiring direct measurements of every network component~\cite{adams2000use, networktomography}. When the available budget does not allow monitors to be deployed at all network nodes, the problem shifts to the optimal placement of a given number of monitors. In already deployed networks, the nodes that can act as monitors may be fixed, whereas in networks under design, monitor locations can be selected to improve tomography performance~\cite{6848079}.

Quantum Network Tomography (QNT) integrates principles from classical network tomography and quantum process tomography to enable path-based estimation of link parameters~\cite{math, starcharacterization, qnt, paper}. 
In QNT, networks are probed through entanglement distribution, generating entangled states encoding error parameters of the traversed links among network monitors, which are capable of performing quantum measurements for parameter estimation. However, path-based inference introduces a fundamental challenge: the multiplicative coupling of link parameters can lead to non-identifiability, where distinct parameter sets produce identical measurement statistics. To address this, monitors must be strategically placed in the network, and probing strategies must be carefully designed to ensure the unique identifiability of all link parameters. Also, to quantify the quality of such estimation, it is essential to evaluate the achievable precision.

Early studies introduced multi-party state-distribution protocols for tree and star networks, considered single-Pauli and depolarizing channels, and employed the Quantum Fisher Information Matrix (QFIM) and Quantum Cramér-Rao Bound (QCRB) to study identifiability and estimation precision~\cite{math, starcharacterization, qnt}. Building on these foundations, our earlier work investigated single-monitor placement in a four-node depolarizing star network~\cite{10821219}, while our recent work compared practical measurement strategies, including local Z-basis, joint Bell-state, and pre-shared entanglement-assisted measurements~\cite{paper}. However, existing studies do not provide a unified framework that jointly addresses link identifiability, monitor placement, link assignment, and monitoring overhead for general network topologies. This work fills this gap by establishing topology-independent learnability and identifiability conditions and formulating monitor placement and link assignment as QFIM-based optimization problems, with both information-maximizing and overhead-constrained formulations for scalable QNT.

Classical network tomography has extensively studied link identifiability and monitor placement from end-to-end measurements, including rank-based identifiability conditions and monitor placement strategies that maximize the number of identifiable links \cite{6848079, 7524374}. These works provide the classical foundation for our formulation. However, QNT introduces additional requirements associated with quantum probe-state generation, quantum measurements, and quantification of parameter estimation information. Our work builds on the classical identifiability and placement principles while incorporating these quantum-specific aspects into a unified monitor placement and measurement assignment framework.

We generalize the QNT framework in which monitors are not limited to end nodes but may be placed anywhere in the network, with measurement ability independent of their position in the network to handle more intricate configurations and make the following key contributions.

\begin{itemize}
    \item We propose a \ac{CSQP} and introduce a fixed-point learnability framework with an explicit algorithm for estimating link-level Werner parameters according to their learning order from monitor measurement outcomes.
    \item We characterize \ac{CSQP} identifiability through the rank of the path-link incidence matrix and prove that the learnability conditions guarantee full rank and a nonsingular \ac{QFIM} for arbitrary network topologies.
    \item We derive closed-form expressions of the Maximum Likelihood Estimators (MLEs) for the link parameters in star networks under a given monitor placement configuration and show that, for an interior feasible joint solution of the considered two-hop indirect probes, the shared-link MLE reduces to the direct-probe estimate.
    \item We address the problem of optimal monitor placement in general quantum network topologies by formulating it as an Integer Linear Program (ILP). Our analysis examines two distinct formulations: (i) monitoring-overhead constrained QFIM formulation (QMF), that maximizes QFIM-trace under monitoring-overhead constraints, and (ii) unconstrained QFIM based formulation (QF) that focuses solely on maximizing QFIM-trace.   
    \item We prove that the proposed optimization constraints enforce \ac{CSQP} learnability and identifiability on arbitrary network topologies without requiring topology-specific reformulation.
    \item Finally, we numerically evaluate the proposed ILP formulations on star and tree networks, analyzing the resulting monitor placement and measurement assignments, and compare direct monitoring-enforced and non-enforced assignments to assess their impact on the QCRB under both QF and QMF.
\end{itemize}

\section*{Methods}
We model the network as a graph $G=(V,E)$, where nodes represent a set of quantum processors and edges $E$ denote the set of quantum links that interconnect the processors in $V$. Nodes in $V$ are divided into \textit{intermediate} and \textit{monitor} nodes. Intermediate nodes can generate bipartite entanglement with their neighbors and perform entanglement swapping operations to generate end-to-end entanglement. Monitor nodes, denoted by $M \subseteq V$, have the same capabilities of intermediate nodes with the additional capability of performing quantum measurements to estimate link parameters. The model assumes that all errors that occur throughout end-to-end entanglement generation appear as quantum state transmission errors over network links~\cite{schumacher1996sending}, such that the entangled states generated through link $e_i \in E$ are corrupted by the two-qubit depolarizing channel
\begin{align}
    \mathcal{E}_{i}(\rho) = w_i\rho + (1-w_i) \frac{\mathbb{I}_4}{4}
\end{align}
of parameter $w_i \in (0, 1)$~\cite{PhysRevA.80.042330}.

Thus, states generated through link $e_i$ are Werner states of form
\begin{align}
    & \rho(w_i) = w_i \, |\Phi^+\rangle \langle \Phi^+| + (1 - w_i)\frac{\mathbb{I}_4}{4}
\end{align}
where $\ket{\Phi^+} = (\ket{00} + \ket{11})/{\sqrt{2}}$ and $\mathbb{I}_4$ is the $4\times 4$ identity matrix. Werner states are commonly used to model entangled states generated through noisy links, considered as a worst-case scenario for practical purposes \cite{barnuM1998information, werner1, werner2} as the channel outputs the maximally mixed state with probability $(1 - w_i)$.

\subsection*{Probe State Distribution and Measurement}
\begin{figure*}[t]
\centering
\subfloat[\label{fig1a}]{\includegraphics[scale=0.35]{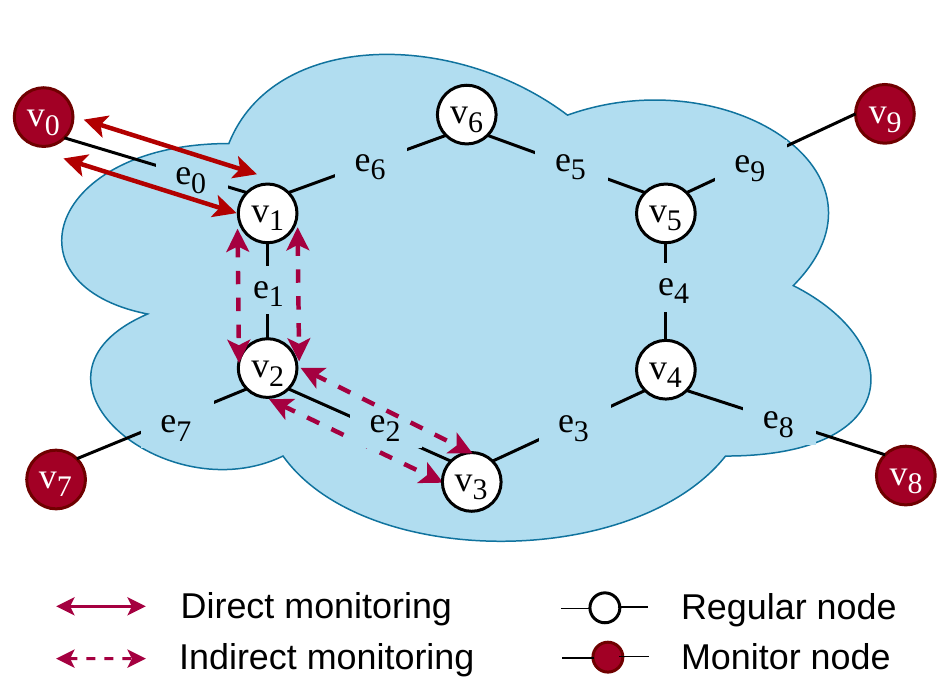}}%
\hfil
\subfloat[\label{fig1b}]{\includegraphics[scale=0.4]{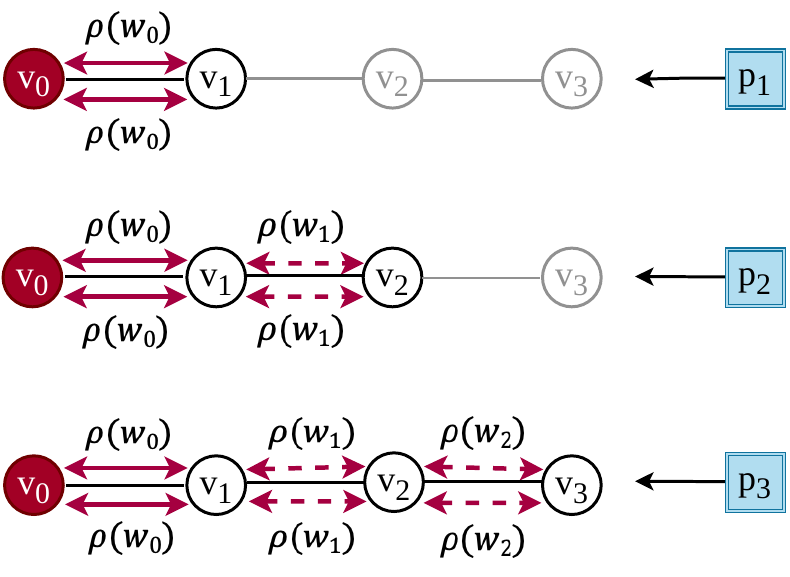}}%
\caption{Probe-state distribution. (a) Link $e_0$ is directly monitored by monitor node $v_0$, shown by solid red arrows. Links $e_1$ and $e_2$ are indirectly monitored by $v_0$, shown by dashed red arrows. (b) Probes $\rho(w_0^2)$, $\rho(w_0^2w_1^2)$, and $\rho(w_0^2w_1^2w_2^2)$ are used to estimate the links $e_0, e_1$ and $e_2$ respectively, along the path $v_0-v_1-v_2-v_3$, and are denoted by $p_1$, $p_2$ and $p_3$, respectively. $p_1$ is generated at $v_0$ by distributing two Werner states between $v_0$ and $v_1$ followed by BSM at $v_1$. $p_2$ is generated at $v_0$ by distributing two Werner states between both $v_0$-$v_1$ and $v_1$-$v_2$ followed by two BSMs at $v_1$ and one BSM at $v_2$. $p_3$ is generated at $v_0$ by distributing two Werner states between $v_0$-$v_1$, $v_1$-$v_2$ and $v_2$-$v_3$ followed by two BSMs at each intermediate node, $v_1$ and $v_2$, and one BSM at $v_3$.}
\label{fig1}
\end{figure*}

Probe state generation has the objective of generating Werner states at monitor nodes for estimation. In particular, it corresponds to the generation of an entangled state in a monitor node that can be measured locally in the monitor. In this framework, each probe is created through link-level entanglement generation and entanglement swapping over network paths that contain one or more network links.

Let $ P_{ab} = [v_a, \ldots, v_b]$ be a path with $l$ links connecting nodes $v_a$ and $v_b$, and assume that a monitor is placed at $v_a$. Each link $e_i = (v_k, v_{j}) \in P_{ab}$ is parameterized by $w_i$, representing the quality of the link through its connection with entanglement fidelity, $f_{\Phi}(w_i) = \frac{1 + 3w_i}{4}$ \cite{paper}. To probe path $P_{ab}$, we generate two copies of the Werner states $\rho(\prod_{i \in P_{ab}} w_i)$ between $v_a$ and $v_b$ through entanglement distribution over the links of $P_{ab}$. The two copies of the end-to-end Werner states are combined through a BSM operation performed at $v_b$, generating the Werner state
\begin{align}
    \rho_{ab} = \rho\!\left( \prod_{i \in P_{ab}} w_i^{\,2} \right) \label{eq:probeState}
\end{align}
at $v_a$. A final BSM is performed locally at $v_a$, yielding classical outcomes that can be used for parameter estimation. Each cyclic probe associated with a path $P_{ab}$ requires two end-to-end entangled states, and therefore involves $2|P_{ab}|$ link-level entanglement
generations, in addition to the required entanglement swapping operations.
Thus, the physical resource cost of a probe increases with its path length. Probe generation is illustrated in Fig. \ref{fig1} for three different paths $P_{01} = [v_0, v_1], P_{02} = [v_0, v_1, v_2],$ and $P_{03} = [v_0, v_1, v_2, v_3]$, which yield the probe states $\rho(w_0^2), \rho(w_0^2w_1^2),$ and $\rho(w_0^2w_1^2w_2^2)$ respectively.

This probing procedure can be visualized as performing quantum communication through network cycles formed by the concatenation of an ordered path $P_{ab}$ with its reversed version $P_{ba}$. In particular, states prescribed by \eqref{eq:probeState} can be seen as the result of end-to-end entanglement generation through a virtual path $P_{ab} \cup P_{ba}$, where the key property is the quadratic dependency of $\rho_{ab}$ with the Werner parameters associated with all the links in $P_{ab}$. Due to this property, we refer to this probing scheme as \textit{cyclic probing} and to states following \eqref{eq:probeState} as \textit{cyclic probes}.

A link $e_i = (v_j, v_k)$ is \textit{directly monitored} if there is a monitor placed either in $v_j$ or $v_k$ and the single-edge path $P_{jk} = e_i$ is used to generate the probe state $\rho_{jk} = \rho(w_i^2)$. Link $e_i$ is \textit{indirectly monitored} through a path $P_{ak} = [v_a,\ldots, v_j, v_k]$ if there is a monitor at node $v_a$ and $P_{ak}$ is used to generate a probe state $\rho_{ak}$ following \eqref{eq:probeState}. Directly and indirectly monitored links are illustrated in Fig. \ref{fig1} with solid and dashed arrows, respectively. 

The construction of probes that generate a parameter-dependent quantum state that is local to a single network node is of key importance. Previous QNT approaches focused on the generation of probe states that spanned multiple monitors (end-nodes) at once~\cite{math, starcharacterization, qnt}, i.e, multipartite entangled states shared by multiple network nodes, or single-qubit probe states distributed to single monitors~\cite{wang2025quantum}. Cyclic probes differ significantly from these other methods, allowing optimal measurements for parameter estimation to be performed without the requirement of pre-shared end-to-end entanglement, a necessary resource for the aforementioned multipartite distribution approaches. Moreover, cyclic probing is the entanglement distribution version of classical cycle-based routing~\cite{gopalan2011identifying, networktomography}, a connection that we later explore in this paper to derive a QNT protocol for networks of arbitrary topologies.

Cyclic probes allow for the estimation of end-to-end Werner parameters for probing paths. In particular, projective measurements of $\rho_{ab}$ allow estimating the end-to-end Werner parameter $w_{ab} = \prod_{i \in P_{ab}} w_i$. Suppose that $N_{ab}$ copies of $\rho_{ab}$ are used to estimate $w_{ab}$. Since a BSM at $v_a$ is prescribed to obtain data for estimation, the estimator
\begin{equation}\label{eq:e2eWernerMLE}
    \hat{w}_{ab} = \sqrt{\frac{4 N_{00}^{ab} - N^{ab}}{3 N^{ab}}},
\end{equation}
where $N_{00}^{ab}$ denotes the number of BSM outcomes corresponding to $\ket{\Phi^{+}}$, is a \ac{MLE} for $w_{ab}$. The prescription of BSMs to extract estimation data from cyclic probes is not arbitrary. In fact, the BSM is an optimal measurement for the estimation of $w_{ab}$ in terms of the \ac{QCRB}. This follows directly from the fact that Werner states are diagonal in the Bell basis~\cite{math,qfim}. Moreover, since the probes are independent, their Fisher information is additive, so the total classical FIM equals the \ac{QFIM} $F_Q^{\mathcal{P}}$ used in this work. Hence, the usual incompatibility of multi-parameter quantum estimation does not arise, and the corresponding \ac{QCRB} is asymptotically attainable.

\subsection*{Quantum Network Tomography}\label{sec:QNT}
QNT protocols provide estimators capable of identifying link-level error parameters of all network links based on data obtained from measurements performed at monitors~\cite{starcharacterization,qnt}. Following the QNT framework, a protocol is described by a set of probe states $\mathcal{P}$, a set of \ac{POVMs} $\Pi$ used to extract information from these probe states, and estimators that combine POVM results to estimate link parameters. We now propose the \ac{CSQP} family, based on cyclic probing and capable of identifying Werner parameters in networks of arbitrary topologies. \ac{CSQP}s are the quantum analogues of the classical cyclic routing network tomography protocols~\cite{networktomography}.

Given a network graph $G$, assume that the monitor nodes $M \subseteq V$ are fixed. For cyclic probing, there exists a one-to-one correspondence between a path $P_{ab}$ and a probe state $\rho_{ab}$. Thus, let $\mathcal{P}$ denote a set of network paths. Every path $P \in \mathcal{P}$ terminates in a unique link, denoted by $\tau(P)$, and we write $P^{-} = P \setminus \{\tau(P)\}$ for the (possibly empty) \textit{prefix} of $P$. A path with $P^{-} = \emptyset$ is a direct probe of $\tau(P)$. Throughout we adopt the standing assumption $w \in (0,1)^{|E|}$, i.e., no link is completely depolarizing. This assumption is necessary rather than technical: a link with $w_l = 0$ makes every path containing it uninformative about all remaining links on that path.

We define learnability constructively as the least fixed point of a monotone set operator. This makes the definition well posed and supplies both an estimation algorithm and an identifiability proof that are valid on arbitrary topologies.

\begin{definition}[Learnable set]\label{def:learnable}
For any subset $L\subseteq E$ of links that have already been learned, let $\Gamma_{\mathcal P}:2^E\rightarrow 2^E$ be the operator
\begin{align}
    \Gamma_{\mathcal{P}}(L) \;=\; \big\{\, \tau(P) \;:\; P \in \mathcal{P},\; P^{-} \subseteq L \,\big\}.
\end{align}
Define $L_{0} = \emptyset$ and $L_{t+1} = L_{t} \cup \Gamma_{\mathcal{P}}(L_{t})$ for $t \geq 0$, where $L_t$ denotes the set of links learned by the end of learning round $t$. The \emph{learnable set} of $\mathcal{P}$ is
\begin{align}
    \mathbb{L}(\mathcal{P}) \;=\; \bigcup_{t \geq 0} L_{t},
\end{align}
and a link $e_i$ is \emph{learnable in $\mathcal{P}$} if $e_i \in \mathbb{L}(\mathcal{P})$. For $e_i \in \mathbb{L}(\mathcal{P})$ we call $t(e_i) = \min\{t \geq 1 : e_i \in L_{t}\}$ the \emph{learning round} of $e_i$.
\end{definition}

The fixed-point construction is equivalent to forward chaining on the definite Horn implications $P^{-}\rightarrow\tau(P)$ associated with the probe paths: once all links in $P^{-}$ are learnable, $\tau(P)$ becomes learnable. Consequently, the learnable set can be computed by standard forward propagation. Moreover, learnability is monotone with respect to the probe set: if $\mathcal{P}\subseteq\mathcal{P}'$, then $\mathbb{L}(\mathcal{P})\subseteq\mathbb{L}(\mathcal{P}')$. Thus, adding probe paths cannot make an already learnable link unlearnable.

Because $\Gamma_{\mathcal{P}}$ is monotone and $E$ is finite, the sequence $(L_{t})_{t \geq 0}$ is non-decreasing and stabilizes after at most $|E|$ steps. Hence $\mathbb{L}(\mathcal{P}) = L_{|E|}$ is the least fixed point of $L \mapsto L \cup \Gamma_{\mathcal{P}}(L)$ and is computable in polynomial time. Crucially, $\mathbb{L}(\mathcal{P})$ is defined from $\mathcal{P}$ alone and makes no reference to the \ac{CSQP} properties, so those properties can be stated in terms of it without circularity.

Link learnability connects with parameter identifiability through the following Lemma.

\begin{lemma}\label{lemma:learnability}
Let $w \in (0,1)^{|E|}$. If $e_i \in \mathbb{L}(\mathcal{P})$, then $w_i$ is identifiable from the measurement statistics of the probes in $\mathcal{P}$, and the corresponding estimator produced by Algorithm~\ref{alg:CSQPestimate} is consistent.
\end{lemma}
\begin{proof}
We argue by strong induction on the learning round $t(e_i)$.

Let $t(e_i) = 1$. Then $e_i = \tau(P)$ for some $P \in \mathcal{P}$ with $P^{-} = \emptyset$, i.e., $P$ is the single-link path $e_i$ and the associated cyclic probe is $\rho(w_i^{2})$. Its Bell-basis outcome probabilities are $(1+3w_i^{2})/4$ for $\ket{\Phi^{+}}$ and $(1-w_i^{2})/4$ for each remaining outcome, which is strictly monotone in $w_i$ on $(0,1)$. Hence $w_i$ is identified and \eqref{eq:e2eWernerMLE} is consistent by the law of large numbers applied to multinomial outcome counts together with the continuous mapping theorem.

Let us assume the claim for every link with the learning round at most $s$ and let $t(e_i) = s+1$. By Definition~\ref{def:learnable} there is a \emph{witness path} $P \in \mathcal{P}$ with $\tau(P) = e_i$ and $P^{-} \subseteq L_{s}$. Therefore, every $e_l \in P^{-}$ satisfies $t(e_l) \leq s$ and, by the induction hypothesis, $w_l$ is identified and consistently estimated. The probe associated with $P$ is $\rho(X_{P})$ with $X_{P} = \prod_{e_l \in P} w_l^{2}$, so $X_{P}$, and hence $w_{P} = \sqrt{X_{P}} = w_i \prod_{e_l \in P^{-}} w_l$, is identified as in the base case. Since $w_l > 0$ for all $l$, the prefix product $w_{P^{-}} = \prod_{e_l \in P^{-}} w_l$ is strictly positive and
\begin{align}
    w_i \;=\; \frac{w_{P}}{w_{P^{-}}} \label{eq:wratio}
\end{align}
identifies $w_i$. Consistency of the plug-in estimator $\hat{w}_i = \hat{w}_{P}/\hat{w}_{P^{-}}$ follows from the consistency of $\hat{w}_{P}$ and of each $\hat{w}_l$, $e_l \in P^{-}$, together with the continuous mapping theorem, the denominator converging to a strictly positive limit.
\end{proof}

The \ac{CSQP} family is composed of QNT protocols with probe sets $\mathcal{P}$ for which the following properties hold:
\begin{enumerate}
    \item all paths in $\mathcal{P}$ start on a monitor node;
    \item for every link $e_i = (v_b, v_k) \in E$ there exists one path that ends with $e_i$, a path which can be the link itself and, therefore every link is either directly or indirectly monitored in $\mathcal{P}$;
    \item every path in $\mathcal{P}$ has a learnable prefix, i.e., $P^{-} \subseteq \mathbb{L}(\mathcal{P})$ for every $P \in \mathcal{P}$, with $\mathbb{L}(\mathcal{P})$ given by Definition~\ref{def:learnable}. Equivalently, $P_{ak} = [v_a,\ldots,v_l,v_b,v_k] \in \mathcal{P}$ only if all links of $P_{ab} = [v_a,\ldots,v_l,v_b]$ are learnable in $\mathcal{P}$.
\end{enumerate}
Properties 2 and 3 together are equivalent to requiring $\mathbb{L}(\mathcal{P})=E$, since every link terminates a path whose prefix is learnable. We retain these properties separately because they explicitly describe the probe structure used by the estimation procedure.
The \ac{CSQP} properties are purely combinatorial: they constrain only the incidence structure of $\mathcal{P}$, i.e., path membership, terminal links, and prefixes, and impose no requirement on the topology of $G$. They are nevertheless sufficient to estimate the Werner parameter of every network link, a result formalized in the following Theorem.

\begin{theorem}\label{th:CSQPIdentifiability}
Let $G=(V,E)$ be an arbitrary network and let $\mathcal{P}$ be a \ac{CSQP} over $G$. Then
\begin{enumerate}
    \item $\mathcal{P}$ contains at least one direct probe;
    \item $\mathbb{L}(\mathcal{P}) = E$, i.e., every link of $G$ is learnable;
    \item Algorithm~\ref{alg:CSQPestimate} terminates after at most $|E|$ passes of its outer loop and returns consistent estimates $\hat{w}_i$ for all $e_i \in E$. The corresponding learning order can be implemented in $O\!\left(\sum_{P\in\mathcal{P}} |P|\right)$ time using unresolved-prefix counters and a worklist of newly learned links.
\end{enumerate}
\end{theorem}
\begin{proof}
\emph{(i)} Suppose, for contradiction, that no path of $\mathcal{P}$ is a direct probe, i.e., $P^{-} \neq \emptyset$ for every $P \in \mathcal{P}$. Then $\Gamma_{\mathcal{P}}(\emptyset) = \emptyset$, hence $L_{t} = \emptyset$ for all $t$ and $\mathbb{L}(\mathcal{P}) = \emptyset$. Property~3 then forces $P^{-} \subseteq \emptyset$ for every $P \in \mathcal{P}$, i.e., every path is a direct probe, contradicting the assumption. Property~2 guarantees $\mathcal{P} \neq \emptyset$ whenever $E \neq \emptyset$, so the argument is not vacuous. Hence $L_{1} \neq \emptyset$.

\emph{(ii)} Let $e_i \in E$ be arbitrary. By Property~2 there exists $P \in \mathcal{P}$ with $\tau(P) = e_i$, and by Property~3 we have $P^{-} \subseteq \mathbb{L}(\mathcal{P})$. Since $\mathbb{L}(\mathcal{P})$ is a fixed point of $L \mapsto L \cup \Gamma_{\mathcal{P}}(L)$,
\begin{align}
    e_i = \tau(P) \in \Gamma_{\mathcal{P}}\big(\mathbb{L}(\mathcal{P})\big) \subseteq \mathbb{L}(\mathcal{P}).
\end{align}
As $e_i$ was arbitrary, $E \subseteq \mathbb{L}(\mathcal{P}) \subseteq E$, which proves the claim. Note that this step is where Property~3 does its work, and that it requires no ordering of the paths.

\emph{(iii)} Algorithm~\ref{alg:CSQPestimate} initializes $L$ with the links terminating direct probes, so after line~6 we have $L = L_{1}$, non-empty by~(i). Each pass of the repeat loop adds to $L$ every link $\tau(P)$ whose prefix is contained in the value of $L$ at the start of that pass; hence after the $t$-th pass $L \supseteq L_{t+1}$, and the loop exits only when no further link can be added, i.e., when $L$ satisfies $L = L \cup \Gamma_{\mathcal{P}}(L)$. Since $(L_{t})$ is non-decreasing and reaches $\mathbb{L}(\mathcal{P})$ after at most $|E|$ steps, the loop performs at most $|E|$ passes and terminates with $L = \mathbb{L}(\mathcal{P})$, which equals $E$ by~(ii). Every link receives its estimate in the pass at which it enters $L$, through the recursion of Lemma~\ref{lemma:learnability}, so consistency of each $\hat{w}_i$ follows from that Lemma. For an efficient implementation, associate with each path $P$ an unresolved-prefix counter initialized to $|P^{-}|$ and maintain a worklist of newly learned links. Whenever a link becomes learnable, the counters of the paths containing that link in their prefixes are decremented. A path becomes eligible when its counter reaches zero. Since each path-link incidence is processed at most once, the learning order can be determined in
$O\!\left(\sum_{P\in\mathcal{P}} |P|\right)$ time.
\end{proof}

    \begin{algorithm}[ht]
        \caption{CSQP estimation (fixed point)}
        \label{alg:CSQPestimate}
        \begin{algorithmic}[1]
            \Require Network $G=(V,E)$; CSQP path set $\mathcal{P}$; probe measurement dataset $N$.
            \Ensure Estimates $\hat{w}_i$ for every $e_i \in E$.

            \State $L \gets \emptyset$;\; $\mathcal{Q} \gets \mathcal{P}$;
            \ForAll{$P \in \mathcal{Q}$ with $P^{-} = \emptyset$} \Comment{direct probes, $P = (e_i)$}
                \State $\hat{w}_i \gets \sqrt{\big(4N_{00}^{P} - N^{P}\big)\big/\big(3N^{P}\big)}$;
                \State $L \gets L \cup \{e_i\}$;\; $\mathcal{Q} \gets \mathcal{Q} \setminus \{P\}$;
            \EndFor
            \Repeat
                \State $\textit{progress} \gets \textsc{false}$;\; $\mathcal{Q}' \gets \mathcal{Q}$;
                \ForAll{$P \in \mathcal{Q}'$ with $\tau(P) = e_i$}
                    \If{$P^{-} \subseteq L$} \Comment{prefix already resolved}
                        \State $\hat{w}_{P^{-}} \gets \prod_{e_l \in P^{-}} \hat{w}_{l}$;
                        \State $\hat{w}_i \gets \sqrt{\big(4N_{00}^{P} - N^{P}\big)\big/\big(3N^{P}\,\hat{w}_{P^{-}}^{2}\big)}$;
                        \State $L \gets L \cup \{e_i\}$;\; $\mathcal{Q} \gets \mathcal{Q} \setminus \{P\}$;
                        \State $\textit{progress} \gets \textsc{true}$;
                    \EndIf
                \EndFor
            \Until{$\textit{progress} = \textsc{false}$}
            \If{$L \neq E$}
                \State \Return \textsc{not identifiable}; \Comment{$\mathcal{P}$ is not a CSQP}
            \EndIf
            \State \Return $\{\hat{w}_i\}_{e_i \in E}$
        \end{algorithmic}
    \end{algorithm}

The learnability framework above can also be connected to a rank-based characterization of identifiability through the path-link incidence matrix. Let $\mathcal{P} = \{P^{1},\ldots,P^{|\mathcal{P}|}\}$ and define the path-link incidence matrix $A \in \{0,1\}^{|\mathcal{P}| \times |E|}$ by $A_{r,l} = 1$ if $e_l \in P^{r}$ and $A_{r,l} = 0$ otherwise. For $w \in (0,1)^{|E|}$ write $u_l = \ln w_l$ and $u = (u_1,\ldots,u_{|E|})^{\top}$.
\begin{proposition}[Log-linear form of cyclic probing]\label{prop:loglinear}
The measurement statistics of the probe associated with $P^{r}$ depend on $w$ only through $X_{P^{r}} = \prod_{e_l \in P^{r}} w_l^{2}$, and
\begin{align}
    \ln X_{P^{r}} = 2\,(Au)_{r}, \qquad r = 1,\ldots,|\mathcal{P}|. \label{eq:loglinear}
\end{align}
Consequently, $w$ is globally identifiable from $\mathcal{P}$ on $(0,1)^{|E|}$ if and only if $\operatorname{rank}(A)=|E|$.
\end{proposition}
\begin{proof}
The probe state is $\rho(X_{P^{r}})$, which is Bell diagonal with spectrum $\{(1+3X_{P^{r}})/4,\,(1-X_{P^{r}})/4,\,(1-X_{P^{r}})/4,\,(1-X_{P^{r}})/4\}$. Every measurement statistic of that state is therefore a function of $X_{P^{r}}$ alone and $X \mapsto (1+3X)/4$ is injective on $(0,1)$, so the statistics of probe $r$ determine $X_{P^{r}}$ and nothing more. Equation~\eqref{eq:loglinear} follows on taking logarithms.
If $\operatorname{rank}(A) = |E|$ then \eqref{eq:loglinear} determines $u$, hence $w$, uniquely. Conversely, if $\operatorname{rank}(A) < |E|$ there exists $z \neq 0$ with $Az = 0$. For $\varepsilon$ small enough that $u + \varepsilon z$ remains in the open set $(-\infty,0)^{|E|}$, the parameter vectors $w$ and $w' = \exp(u + \varepsilon z)$ are distinct yet give $X_{P^{r}} = X'_{P^{r}}$ for every $r$, and therefore identical statistics for every probe of $\mathcal{P}$. Hence $w$ is not identifiable.
\end{proof}

Proposition~\ref{prop:loglinear} recasts cyclic-probing \ac{QNT} as a linear inverse problem in $\ln w$, in direct analogy with multiplicative classical network tomography~\cite{networktomography}, and reduces identifiability to a rank condition that is well defined for any graph.
\begin{proposition}\label{th:rank}
Let $G=(V,E)$ be an arbitrary network and $\mathcal{P}$ a \ac{CSQP} over $G$. Then $\operatorname{rank}(A) = |E|$. In particular, $\mathcal{P}$ identifies all link Werner parameters on $(0,1)^{|E|}$, and $|\mathcal{P}| \geq |E|$, with equality attainable.
\end{proposition}

\begin{proof}
By Theorem~\ref{th:CSQPIdentifiability}(ii), $\mathbb{L}(\mathcal{P}) = E$, so every link $e_i$ has a finite learning round and admits a witness path $P_{(i)}$ terminating in $e_i$. Order the links by nondecreasing learning round and order the corresponding witness rows in the same way. Since every link in $P_{(i)}^{-}$ has a strictly smaller learning round than $e_i$, the resulting $|E|\times |E|$ submatrix of $A$ is lower triangular with unit diagonal. Hence $\operatorname{rank}(A)=|E|$. Identifiability follows from Proposition~\ref{prop:loglinear}, and $|\mathcal{P}| \geq |E|$. Equality holds when $\mathcal{P}$ consists exactly of the witness paths.
\end{proof}

\subsection*{Estimation efficiency and the QFIM}
We use the \ac{QFIM} and the \ac{QCRB} to quantify the achievable precision in estimating link parameters with cyclic probing and the \ac{CSQP} family. The QFIM of a parameter-dependent state $\sigma(w)$ quantifies the amount of information the state carries about the parameter vector $w$. Obtained from the inverse of the \ac{QFIM}, the \ac{QCRB} establishes a lower bound on the covariance matrix of any unbiased estimator $\hat{w}$ of $w$ that can be obtained from measurements of $\sigma(w)$~\cite{helstrom1969quantum, qfim}. Let $F_{Q}[\sigma(w)]$ denote the QFIM of the parameter-dependent state $\sigma(w)$ with respect to the link-parameter vector $w$. Since the QFIM is additive with respect to tensor products, the QFIM of a \ac{CSQP} is the entry-wise sum of the QFIMs for each individual probe state used in \ac{QNT}~\cite{starcharacterization}. Thus, the \ac{QFIM} $F_{Q}^{\mathcal{P}}$ of the \ac{CSQP} specified by $\mathcal{P}$ with measurement dataset $N$ is 
\begin{align}
    & F_{Q}^{\mathcal{P}}=\sum_{P_{ab} \in \mathcal{P}} N_{ab} F_{Q}[\rho(w_{ab}^{2})],
\end{align}
where $N_{ab}$ is the number of measurements of probes $P_{ab}$.

The \ac{QFIM} $F_Q[\rho(w_{i}^2)]$ of probe $\rho(w_{i}^2)$ representing the direct monitoring of link $e_i$ is a diagonal matrix with a single non-zero element given by
\begin{equation}\label{eq:directQFI}
    F_{ii}^{D} = \frac{12 w_i^2}{(1 + 3 w_i^2)(1 - w_i^2)} .
\end{equation}
For indirect monitoring of link $e_i = (v_j,v_k)$ through path $P_{ak} \in \mathcal{P}$, the entries of $F_{Q}[\rho(w_{ak}^{2})]$ are
\begin{equation}
\label{eq:indirectQFI}
    F_{ig}^{I} = \frac{12 w_i w_g \bigg(\prod_{\substack{{l \in P_{ak}}\\l \ne i}} w_l^2 \prod_{\substack{{l \in P_{ak}}\\l \ne g}} w_l^2\bigg)}{(1 + 3 \prod_{l \in P_{ak}} w_l^2)(1 - \prod_{l \in P_{ak}}w_l^2)}.
\end{equation}
Complete derivations, together with numerical evaluations for a four-node star network, are provided in Appendix~A.

Expressions~\eqref{eq:directQFI} and~\eqref{eq:indirectQFI} admit a compact matrix form that makes the connection with Proposition~\ref{th:rank} explicit. Since the probe state associated with $P \in \mathcal{P}$ depends on $w$ only through $X_{P} = \prod_{l \in P} w_l^{2}$, its \ac{QFIM} is a rank-one matrix,
\begin{align}
    F_{Q}[\rho(X_{P})] = c(X_{P})\,\nu_{P}\nu_{P}^{\top},
    \qquad
    c(X) = \frac{12X^{2}}{(1+3X)(1-X)}, \label{eq:rankone}
\end{align}
where $(\nu_{P})_{l} = 1/w_l$ if $e_l \in P$ and $0$ otherwise. Equation~\eqref{eq:indirectQFI} is recovered from \eqref{eq:rankone} on writing $\prod_{l \in P,\, l \neq i} w_l^{2} = X_{P}/w_i^{2}$, and \eqref{eq:directQFI} is the special case $P = \{e_i\}$. The rank-one structure is a direct consequence of the fact that a single cyclic probe carries information about one scalar function of $w$ only; identifiability of the full parameter vector is therefore never a property of an individual probe, but of the incidence structure of $\mathcal{P}$ as a whole.

\begin{corollary}[Non-singularity of the QFIM]\label{cor:qfim}
Let $N_{P}$ denote the number of measurements of probe $P$, $D_{w} = \operatorname{diag}(w_1^{-1},\ldots,w_{|E|}^{-1})$ and $C = \operatorname{diag}\big(N_{P}\,c(X_{P})\big)_{P \in \mathcal{P}}$. Then
\begin{align}
    F_{Q}^{\mathcal{P}} = D_{w}\,A^{\top} C\, A\, D_{w},
    \qquad
    \operatorname{rank}\big(F_{Q}^{\mathcal{P}}\big) = \operatorname{rank}(A).
\end{align}
Consequently, for any network $G$ and any \ac{CSQP} $\mathcal{P}$ over $G$ with $N_P \geq 1$, the matrix $F_Q^{\mathcal{P}}$ is positive definite and the QCRB $(F_Q^{\mathcal{P}})^{-1}$ is well defined. Moreover, since the probe states are Bell diagonal with parameter independent eigenvectors and are measured in the Bell basis, the corresponding QCRB is asymptotically attainable.
\end{corollary}
\begin{proof}
Write $\nu_{P} = D_{w}a_{P}$, where $a_{P}^{\top}$ is the row of $A$ associated with $P$. Summing \eqref{eq:rankone} over $\mathcal{P}$ with multiplicities $N_{P}$ gives
\begin{align}
    F_{Q}^{\mathcal{P}} = D_{w}\Big(\textstyle\sum_{P \in \mathcal{P}} N_{P}\,c(X_{P})\,a_{P}a_{P}^{\top}\Big)D_{w} = D_{w}A^{\top}CAD_{w}.
\end{align}
For $w \in (0,1)^{|E|}$ the matrix $D_{w}$ is non-singular, and $c(X) > 0$ for $X \in (0,1)$ makes $C$ positive definite, so $A^{\top}CA = (C^{1/2}A)^{\top}(C^{1/2}A)$ has the same rank as $A$. Proposition~\ref{th:rank} gives $\operatorname{rank}(A) = |E|$ for any \ac{CSQP} on any topology, hence $F_{Q}^{\mathcal{P}} \succ 0$.
\end{proof}

Positive definiteness guarantees identifiability, but estimation precision also depends on the conditioning of $F_Q^{\mathcal{P}}$. In particular, longer probing paths and lower-quality links reduce $X_P=\prod_{e_l\in P}w_l^2$ and consequently the QFIM contribution $c(X_P)$. Such paths can therefore reduce the smallest eigenvalue of the total QFIM and increase $\operatorname{Tr}((F_Q^{\mathcal{P}})^{-1})$.
Corollary~\ref{cor:qfim} shows that the identifiability condition of Proposition~\ref{th:rank} and the nonsingularity of the \ac{QFIM} are one and the same condition. Every feasible solution of the optimization formulations introduced below therefore admits a finite \ac{QCRB}, on any topology, with no additional topology-specific requirement.

\subsection*{Maximum Likelihood Estimators for Star Network}
The \ac{MLE} of the Werner parameter vector maximizes the likelihood of the observed measurement outcomes~\cite{fiuravsek2001maximum,qfim} and is asymptotically efficient: its covariance matrix approaches the \ac{QCRB} limit in the large-sample regime. We now describe closed form expressions for the \ac{MLE} of Werner parameters in an $n$-node star network, using the same path notation as in the \ac{CSQP} construction. 

For a directly monitored link $e_q = (v_a, v_h)$, that is not part of any indirect monitoring probe, where $v_a \in M$ is a monitor and $v_h$ is the center of the star, the MLE is given by \eqref{eq:e2eWernerMLE}. Indirect probes in a star are two-hop paths of the form $P_{aj}=[v_a, v_h, v_j]$, with $v_a \in M$ and end-to-end Werner parameter $w_{aj}=w_qw_j$. For an indirectly monitored link $e_j=(v_h,v_j)$, let
$\mathcal{I}_q=\{j: e_j \text{ is indirectly monitored through a path containing } e_q\}$.
The Werner parameter MLE $\hat{w}_j$ is obtained from the MLE
$\hat{w}_q$ of the parameter associated with the directly monitored link $e_q$,
and is given by
\begin{equation}
\label{eq:indirect_star_mle}
    \hat{w}_j
    =
    \frac{\hat{w}_{aj}}{\hat{w}_q}
    =
    \sqrt{
    \frac{4N_{00}^{aj}-N^{aj}}
    {3N^{aj}\hat{w}_q^2}
    },
    \qquad j\in\mathcal{I}_q .
\end{equation}

In the case of a directly monitored link $e_q$ that is part of indirect probes, the MLE $\hat{w}_q$ is obtained from the joint likelihood of all probe measurements that utilize $e_q$. For an interior feasible joint MLE satisfying $0 < \hat{x}_{aj} < \hat{w}_q^2$ for every $j\in\mathcal{I}_q$, the indirect likelihood contributions vanish, and the estimate of the common link reduces to
\begin{equation}
\label{eq:neweqmle}
    \hat{w}_q
    = \sqrt{\frac{4N_{00}^{ah} - N^{ah}}{3N^{ah}}}.
\end{equation}
The detailed derivations of the \ac{MLE} equations with the corresponding \ac{MSE} convergence results are presented in Appendix~B.

The above closed-form expressions are specific to the star-network setting, where each indirect probe contains one directly monitored link and one indirectly monitored link. Although the likelihood equations are generally coupled through the products of link parameters appearing in indirect probe paths, for an interior feasible joint solution this coupling collapses in the considered star network setting, as shown in Appendix B. For general multi-link paths, however, such cancellation need not occur, and the likelihood equations may form a coupled nonlinear system involving several link parameters. Deriving compact closed-form MLEs in that setting is therefore considerably more difficult. Hence, for arbitrary network topologies, we use the CSQP estimator described in Algorithm~\ref{alg:CSQPestimate}.
\vspace{0.5em}

\section*{Results}

\subsection*{The Optimal Monitor Placement Problem}
The optimal monitor placement problem consists of placing monitors in the network nodes and specifying QNT probes such that: (i) error parameters for all network links can be estimated and (ii) the solution obtained optimizes a desired measure of estimation performance.
We study the optimal monitor placement problem through an \ac{ILP} framework that allows exploring principles of optimal experimental design together with monitoring load balancing, which is of practical interest.

\subsubsection*{Experimental Design Principles}
Common approaches in experimental design include specifying probes that minimize the average variance (A-optimality), maximize the QFIM determinant (D-optimality), maximize the QFIM trace (T-optimality), and maximize the minimum QFIM eigenvalue (E-optimality)~\cite{3,rady2009relationships}, each yielding distinct QNT solutions. From a statistical standpoint, there are strong reasons to prefer objectives that focus on the inverse or the determinant of the QFIM. In particular, A-optimality, i.e., minimizing the trace of the inverse QFIM,
directly targets the aggregate estimation error through the QCRB, since $\mathrm{Tr}(\mathrm{F}_Q^{-1})$ lower bounds the sum of the variances of any locally unbiased estimator obtained from probe measurements. Likewise, D-optimality encourages designs that are informative in all parameter directions by penalizing poorly conditioned matrices. In both cases, these objectives provide robust, theoretically grounded measures of statistical performance. However, these are difficult to optimize analytically and numerically. Both matrix inversion and determinant computations introduce nonlinear and nonconvex expressions in the decision variables, which makes the resulting problem extremely hard to solve when monitor placement decisions are binary. Incorporating these operations into an optimization formulation either renders the model intractable or requires relaxations that undermine the exactness of the solution.

In this work, we maintain scalability by adopting the T-optimal criterion \cite{atkinson2007optimum}, which focuses on maximizing the trace of the QFIM, i.e., the sum of the QFI for all link parameters. The trace grows with the total amount of information contributed by all measurements and, importantly, depends linearly on the variables that encode the monitor placement problem. This allows formulating the monitor placement problem as an ILP since nonlinear constraints are absent. Although T-optimality is a simpler metric than A- or D-optimality, it is widely used as a practical surrogate because it tends to improve estimation precision while keeping the optimization problem manageable \cite{VICKERS2021649, 6194279, ruppert2012optimal}. Nonetheless, T-optimality does not directly quantify estimation error. Moreover, the QFIM contribution $c(X)=12X^{2}/[(1+3X)(1-X)]$ increases without bound as $X\rightarrow1^{-}$. Consequently, the T-optimal objective can be dominated by high-quality, short probing paths. This differs from the QCRB metric $\operatorname{Tr}(F_Q^{-1})$, which is strongly affected by poorly estimated parameter directions. Therefore, maximizing $\operatorname{Tr}(F_Q)$ does not necessarily minimize $\operatorname{Tr}(F_Q^{-1})$. Thus, after solving the ILP using this surrogate objective, we evaluate the resulting designs using the QCRB, which provides a more meaningful measure of practical estimation performance.

The QFIM coefficients depend on the link Werner parameters, which are assumed to be available from prior calibration. Developing an adaptive QNT framework that jointly updates parameter estimates and monitor placement is left for future work.

\begin{table}[h]
\caption{List of Variables}\label{tab1}
\begin{tabularx}{\columnwidth}{@{}lX@{}}
\toprule
Variables & Description \\
\midrule
$V$    & Set of network nodes (vertices) \\
$E$    & Set of network links (edges)   \\
$M$    & Set of monitor nodes   \\
$\mathcal{J}$ & Set of monitor indices \\
$F_{ii}^{D}$    & Diagonal QFIM entry associated with link $e_i$  \\
$F_{gg}^{I}$    & Diagonal QFIM entry associated with link $e_g$  \\
$z_{i,j,k}$    & Link $i$ is indirectly measured by monitor $j$ placed at node k \\
$x_{i,j}$    & Link $i$ is directly measured by monitor $j$ \\
$p_{i,j}$    & Link $i$ is indirectly measured by monitor $j$   \\
$m_{k,j}$    & Monitor $j$ is placed at node $k$  \\
$\lambda_{i}$    & Learning level of link $i$; enforces a well-founded resolution order \\
$P_{k,i}$    &Candidate probing path from node $k$ terminating in link $i$ \\
$\ell_j$    & Number of links measured by monitor $j$   \\
${L_j}^*$    & Monitoring capacity of monitor $j$   \\
\bottomrule
\end{tabularx}
\end{table}

\subsubsection*{QMF and QF ILP Formulations}
We propose two ILP formulations for the optimal monitor placement problem: one that maximizes the trace of the QFIM (QF); and another that obtains the maximum achievable QFIM trace under constraints on the number of links monitored by each network monitor (QMF). The first solution prioritizes estimation information, while the second solution balances estimation information with monitoring overhead and parallel probe generation, i.e., the simultaneous generation of probe states through multiple paths. The variables used in the ILP formulations are listed in Table~\ref{tab1}.

We present the QMF \ac{ILP} formulation and obtain the QF formulation by removing the monitoring capacity constraints from QMF. Both formulations solve the monitor placement problem under the assumption that paths between network nodes are specified \textit{a priori} by a routing algorithm. Specifically, for a given network $G=(V,E)$, for every candidate monitor location $k\in V$ and every target link $e_i=(u,v)\in E$, we construct a
candidate probing path $P_{k,i}$ by selecting a hop-shortest path from $k$ to either endpoint of $e_i$ without traversing $e_i$, and then appending $e_i$ as the final link. If multiple paths have the same minimum length, ties are resolved deterministically. If $k$ is an endpoint of $e_i$, the prefix is empty and $P_{k,i}=\{e_i\}$. We denote the set of all candidate paths $P_{k,i}$ by $\mathcal{P}^{\star}$.

The goal of QMF is to maximize the QFIM-trace subject to per-monitor capacity constraints. To capture this trade-off, the objective function maximizes the trace of the QFIM subject to monitor-specific capacity constraints. The capacity of a monitor corresponds to the maximum number of links the monitor can measure. Let $\mathcal{J}=\{1,\ldots,m\}$ denote the index set of the $m$ available monitors.

Direct measurements of a link can only occur when a monitor is located at one end node of that link. Indirect measurements require a probe path that starts in a monitor node and traverses the target link, e.g., in a star topology, this corresponds to a two-hop path through the hub.
The decision variable $x_{i,j} \in \{0, 1\}$ denotes whether monitor $j$ directly measures link $i$, while $p_{i,j} \in \{0, 1\}$ indicates indirect measurement of link $i$ by monitor $j$. Also, the placement of the monitor $j$ on node $k$ is represented by the decision variable $m_{k,j}$.

To compute the trace of the final QFIM, we sum the diagonal elements of the QFIM corresponding to direct \eqref{eq:directQFI} and indirect \eqref{eq:indirectQFI} monitoring. Each term in the sum is weighted by the corresponding decision variables to reflect the actual monitor placement and measurements. We assume that every selected probe is measured the same number of times. Since this common measurement count scales all QFIM contributions equally, it does not affect the optimal placement and is therefore omitted from the objective function. Here, optimality is defined with respect to the precomputed candidate path set $\mathcal{P}^{\star}$. The measurement overhead of a monitor is the total number of links that are either directly or indirectly measured by the monitor, and the monitoring-overhead is the maximum measurement overhead among all deployed monitors. 

The QMF optimization formulation is defined as follows:
\begin{align}
\label{obj}
    \max_{\substack{x_{i,j},\,p_{i,j},\,m_{k,j},\,z_{i,j,k}}}
\Bigg[
&\sum_{\substack{i\in E\\ j\in \mathcal{J}}}
x_{i,j}F_{ii}^{(D)} +
\sum_{\substack{i\in E\\ j\in \mathcal{J}\\ k\in V}}
z_{i,j,k}
\sum_{g\in P_{k,i}}
F_{gg}^{(I)}(P_{k,i})
\Bigg].
\end{align}
\text{s.t.}
\begin{align}
    & \sum_{j \in \mathcal{J}} (x_{i,j} + p_{i,j}) = 1, \quad \forall i \in E \label{dir-ind}; \quad \\
    & x_{i,j} \leq m_{u, j} + m_{v, j}, 
    \quad \forall i = (u, v) \in E, \forall j \in \mathcal{J} \label{dirr}; \quad \\
    & p_{i,j} + m_{k,j} \leq 1, \qquad \forall i = (u, v) \in E,\; \forall j \in \mathcal{J},\; \forall k \in \{u, v\} \label{indirr}; \quad \\
    & \sum_{k \in V} m_{k,j} = 1, \quad \forall j \in \mathcal{J} \label{mk1}; \quad \\
    & \sum_{j \in \mathcal{J}} m_{k,j} \leq 1, \quad \forall k \in V \label{mk2}; \quad \\
    & z_{i,j,k} \leq p_{i,j} , \quad z_{i,j,k} \leq m_{k,j}, \quad z_{i,j,k} \geq p_{i,j} + m_{k,j} - 1,\nonumber \\
    &\hspace{1.4cm} \forall i\in E,\ \forall j\in\mathcal{J},\ \forall k\in V; \label{zijk}\\
    &\lambda_{h} + 1 \;\leq\; \lambda_{i} + |E|\big(2 - p_{i,j} - m_{k,j}\big), \nonumber \\
    &\hspace{1.4cm} \forall i \in E,\; \forall j \in \mathcal{J},\; \forall k \in V,\; \forall h \in P_{k,i}\setminus\{i\};\label{lambh}\\
    &1 \leq \lambda_{i} \leq |E|, \quad \forall i \in E;\label{lambh2}\\
    & \ell_j = \sum_{i \in E} (x_{i,j} + p_{i,j}) \quad \forall j \in \mathcal{J}; \label{over1} \\
    & \ell_j \; \in \; \{1, \ldots, L_{j}^{\star}\}, \quad \forall\, j \in \mathcal{J}; \label{over2}\\
    & z_{i,j,k} \in \{0,1\}, \quad \forall i \in E, \, j \in \mathcal{J}, \, k \in V; \label{zijk2} \\
    & x_{i,j} \in \{0,1\}, \quad \forall i \in E, \, j \in \mathcal{J}; \label{dirr2} \\
    & p_{i,j} \in \{0,1\}, \quad \forall i \in E, j \in \mathcal{J}; \label{indirr2} \\
    & m_{k,j} \in \{0,1\}, \quad \forall k \in V, j \in \mathcal{J}; \label{mk3} \\
    &\lambda_{i} \in \{1,\ldots, |E|\}, \quad \forall i \in E \label{lambh3};
\end{align}

Constraint \eqref{dir-ind} ensures that each link is measured exactly once, either directly or indirectly by a monitor. Constraint \eqref{dirr} allows for direct measurement of a link only if a monitor is placed at one of the endpoints of that link. Constraint \eqref{indirr} ensures that a link cannot be assigned to indirect monitoring by a monitor located at either of its endpoint nodes. Following Definition \ref{def:learnable}, a link is learnable if its Werner parameter can be estimated from the probe states used for tomography, which by Definition~\ref{def:learnable} requires a well-founded order in which link parameters are resolved. Constraint \eqref{mk1} assigns each monitor to one node, while constraint \eqref{mk2} limits each node to have at most one monitor. Constraint \eqref{zijk} ensures that the auxiliary binary variable $z_{i,j,k}$ is equivalent to the product $p_{i,j}m_{k,j}$, thereby preserving the linearity of the formulation. Let $\lambda_i \in \{1, \cdots, |E|\}$ denote the learning level of link $e_i$. Constraints~\eqref{lambh} and~\eqref{lambh2} encode that order directly through a level, or potential, variable $\lambda_i$ attached to each link. Constraint~\eqref{lambh} is a big-$M$ inequality that becomes active exactly when monitor $j$ is placed at node $k$ and measures link $i$ indirectly, in which case it forces every link $h$ on the prefix of the probing path $P_{k,i}$ to carry a strictly smaller level than $i$; when $p_{i,j}=0$ or $m_{k,j}=0$ the right-hand side exceeds $|E|+1$ and the inequality is slack, so exactly one value of $k$ is binding for each monitor by Constraint~\eqref{mk1}. Constraint~\eqref{lambh2} bounds the levels, which is without loss of generality because Theorem~\ref{th:CSQPIdentifiability} shows that at most $|E|$ learning rounds are ever needed. A link that is assigned the smallest level cannot be measured indirectly, since that would require a link of strictly smaller level; it must therefore be measured directly, which seeds the resolution order. Theorem~\ref{th:ilpTopology} below shows that this is exactly what is needed, on any topology.

Constraints \eqref{lambh} and \eqref{lambh2} require only that the links of the prefix be resolved \emph{before} link $i$, and not that they be resolved by the same monitor. They therefore allow a link $i \in E$ to be indirectly monitored by monitor $j$ even if the links in the path between the node containing $j$ and $i$ are indirectly monitored by monitors different than $j$. If this behavior is not desired, the constraint

\begin{equation}
\begin{aligned}
p_{i,j} \le \frac{1}{|P_{k,i}|} \sum_{h \in P_{k,i}} (x_{h,j}+p_{h,j}) + (1 - m_{k,j}), \forall i \in E,\; \forall j \in \mathcal{J},\; \forall k \in V
\end{aligned}
\label{new}
\end{equation}

replaces Constraints \eqref{lambh} and \eqref{lambh2}. Constraint \eqref{new} ensures that link $i$ can be indirectly monitored by $j$ only if all links along the path connecting $j$ to $i$ are measured either directly or indirectly by $j$. Constraint \eqref{new} is stronger than \eqref{lambh}--\eqref{lambh2}: it requires all links in the prefix of an indirectly monitored path to be assigned to the same monitor. Under the candidate-path construction used here, this induces a valid level ordering satisfying Constraints~\eqref{lambh}--\eqref{lambh2}, so the results below apply to both variants.

In addition to the constraints above, we optionally impose a direct-monitoring enforcement constraint,

\begin{equation}
    2\sum_{j\in \mathcal{J}}x_{i,j} \geq \sum_{j\in \mathcal{J}}\left(m_{u, j} + m_{v, j}\right),\quad \forall i = (u, v) \in E; \label{9}
\end{equation}

This constraint requires a link to be directly monitored whenever a monitor is placed at either of its endpoint nodes. It is an optional design restriction on the measurement assignment and it is not required for CSQP learnability or identifiability. We later compare solutions obtained with and without this optional constraint.

We now establish that the constraint set requires no topological information about $G$ and, in particular, that no additional or modified constraints are needed when moving from stars and trees to arbitrary graphs.

\begin{theorem}[The constraint set is topology independent]\label{th:ilpTopology}
Let $G = (V,E)$ be an arbitrary network, let $\mathcal{P}^{\star}$ be the candidate path set constructed above, and let $(z, x,p,m,\ell,\lambda)$ be any feasible solution of \eqref{dir-ind}--\eqref{lambh3}. Let $k_{j}$ denote the node hosting monitor $j$ and let
\begin{align}
    \mathcal{P} = \big\{\, P_{k_{j},i} \;:\; j \in \mathcal{J},\; i \in E,\; x_{i,j} + p_{i,j} = 1 \,\big\}
\end{align}
be the induced path set. Then $\mathcal{P}$ is a \ac{CSQP} over $G$, and consequently $\mathbb{L}(\mathcal{P}) = E$, $\operatorname{rank}(A) = |E|$ and $F_{Q}^{\mathcal{P}} \succ 0$. 
\end{theorem}
\begin{proof}
We verify the three \ac{CSQP} properties.

\emph{Property 1 (paths start at monitors)}: Every $P_{k,i} \in \mathcal{P}^{\star}$ starts at a candidate location $k$ by construction. Constraint~\eqref{dirr} activates $x_{i,j}$ only when a monitor is placed at an endpoint of $i$, and Constraint~\eqref{lambh} is binding only for the $k$ with $m_{k,j}=1$, which is unique by Constraint~\eqref{mk1}. Every path of $\mathcal{P}$ therefore starts at a monitor node.

\emph{Property 2 (every link terminates some path)}: Constraint~\eqref{dir-ind} assigns each link $i \in E$ to exactly one monitor, either directly or indirectly. Hence $\mathcal{P}$ contains exactly one path terminating in $i$ for every $i \in E$, namely $P_{k_{j},i}$, reducing to the single-link path when the assignment is direct. In particular, $|\mathcal{P}| = |E|$ and
$\tau:\mathcal{P}\rightarrow E$, which maps each path $P\in\mathcal{P}$ to its terminal link $\tau(P)\in E$, is a bijection.

\emph{Property 3 (learnable prefixes)}: We show by strong induction on $\lambda_i$ that $e_i \in \mathbb{L}(\mathcal{P})$ for every $i \in E$. Let $i \in E$ and assume the claim for every link of strictly smaller level. If $i$ is measured directly then, by Property 2, the path of $\mathcal{P}$ terminating in $i$ is the single link $i$, so $P^{-} = \emptyset$ and $e_i \in \Gamma_{\mathcal{P}}(\emptyset) \subseteq \mathbb{L}(\mathcal{P})$. If $i$ is measured indirectly by monitor $j$, then $p_{i,j} = 1$ and $m_{k_{j},j} = 1$, so Constraint~\eqref{lambh} yields $\lambda_h + 1 \leq \lambda_i$, i.e. $\lambda_h < \lambda_i$, for every $e_h \in P_{k_{j},i}^{-}$. By the induction hypothesis $e_h \in \mathbb{L}(\mathcal{P})$ for all such $h$, whence

\begin{align}
    e_i = \tau\big(P_{k_{j},i}\big) \in \Gamma_{\mathcal{P}}\big(\mathbb{L}(\mathcal{P})\big) \subseteq \mathbb{L}(\mathcal{P}).
\end{align}

The induction is well founded: $\lambda$ takes finitely many values on $E$, and a link attaining the minimum of $\lambda$ cannot be measured indirectly, since that would require a link of strictly smaller level. Hence $\mathbb{L}(\mathcal{P}) = E$ and, a fortiori, $P^{-} \subseteq \mathbb{L}(\mathcal{P})$ for every $P \in \mathcal{P}$, which is Property 3. Theorem~\ref{th:CSQPIdentifiability}, Proposition~\ref{th:rank} and Corollary~\ref{cor:qfim} then apply.

\emph{Topology independence}: Constraints~\eqref{dir-ind}--\eqref{lambh3} are expressed solely in terms of $E$, $V$, $\mathcal{J}$, the endpoint map $i \mapsto \{u, v\}$, and the path sets $P_{k,i}$. All of these objects are defined for an arbitrary graph, and no constraint refers to degree, girth, acyclicity, planarity or any other topological invariant.
\end{proof}

Finally, Constraint \eqref{over1} defines the measurement overhead $\ell_j$ for each monitor $j$ as the total number of links assigned to it for both direct and indirect measurement, and Constraint \eqref{over2} upper-bounds the measurement overhead of monitor $j$ by its monitoring capacity $L_j^{\star}$. This formulation allows monitors to have different capacities, which is of practical interest. In the remainder of this work, we assume all monitors have the same capacity, i.e., $L_{j}^{\star} =  L^{*}$. From the definition of monitoring overhead, a necessary condition for
feasibility is $\lceil |E| / m \rceil \leq L^{\star} \leq |E|.$ The lower bound follows from assigning the $|E|$ links among $m$ monitors, but satisfying this bound alone does not guarantee feasibility because of the remaining placement and learnability constraints. This formulation supports parallel probe generation when $L^{\star} < |E|$ and distributes the measurement overhead across monitors.

Assuming that $\mathcal{P}^{\star}$ is previously computed significantly simplifies the ILP formulation since, once monitors are fixed, the paths between them and all network links are also fixed. The QMF formulation specifies a CSQP defining set of probes $\mathcal{P} \subset \mathcal{P}^{\star}$, which is constructed as follows. For every monitor $j \in \mathcal{J}$, let $E_{j}^{D} = \{i: x_{i, j}=1, i \in E\}$ and $E_{j}^{I} = \{i: p_{i,j} = 1, i \in E \}$ denote the set of links directly and indirectly monitored by monitor $j$, respectively. Let $k_j \in V$ denote the node that received monitor $j$, i.e., $k_j$ is the unique node in $V$ for which $m_{k_j,j} =1$. $\mathcal{P}$ is formed by all paths $P_{k_j, i}$ for every monitor $j$, and for every link $i \in E_{j}^D \cup E_{j}^{I}$.

The objective of the QF optimization problem is to maximize QFIM-trace and is obtained by removing constraints \eqref{over1} and \eqref{over2} from the QMF formulation. In this case, the optimization problem focuses exclusively on the T-optimality criterion from experimental design, which often leads to monitor overloading scenarios where there is an imbalance in the number of links handled by each monitor.

\subsection*{Numerical Analysis of Optimal Monitor Placement}
\begin{figure*}[t]
\centering
\subfloat[\label{fig11a}]{\includegraphics[scale=0.5]{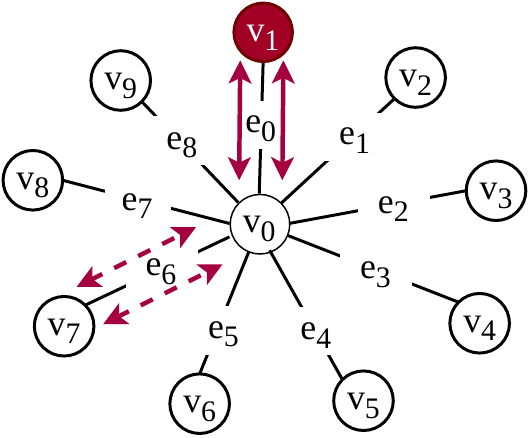}}%
\hfil
\subfloat[\label{fig11b}]{\includegraphics[scale=0.5]{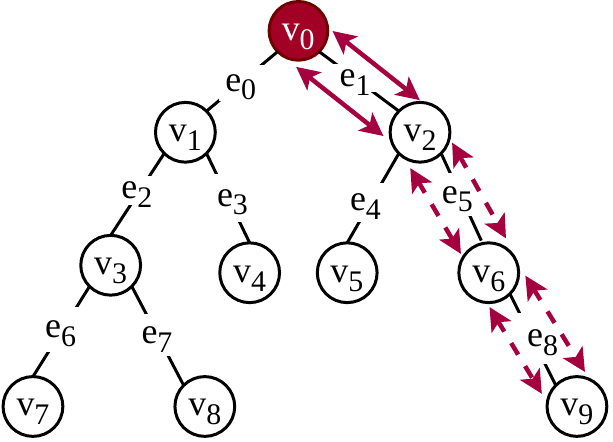}}%
\caption{Network topologies considered in the numerical experiments: (a) $10$-node star network, and (b) $10$-node tree network. A single-monitor deployment and its possible associated measurements along a path are shown as an example.
}
\label{fig11}
\end{figure*}
We study the applicability of our optimization formulations by analyzing ten-node star and tree networks shown in Fig. \ref{fig11} with inhomogeneous noise depicted in Figs. \ref{fig_sim_results1}, \ref{fig_sim_results1a}, and \ref{fig_sim_results4}, \ref{fig_sim_results4a}, respectively. We solve both the QMF and QF formulations varying the number of monitors from one to nine in the case of the star, and from one to four in the case of the tree. The maximum number of monitors in both cases allows for full direct monitoring coverage, i.e., every link can be directly monitored. Moreover, we solve the QMF for both the star and the tree network assuming that the monitoring capacity is $\lceil |E| / m \rceil $ for each monitor when $m$ monitors are used. We solve both QF and QMF with and without direct monitoring enforcement constraint (\ref{9}) to evaluate its effect on estimation performance and monitor overhead.

\subsubsection*{Star Networks}
\begin{figure*}[t]
\centering
\subfloat[\label{fig2a}]{\includegraphics[scale=0.25]{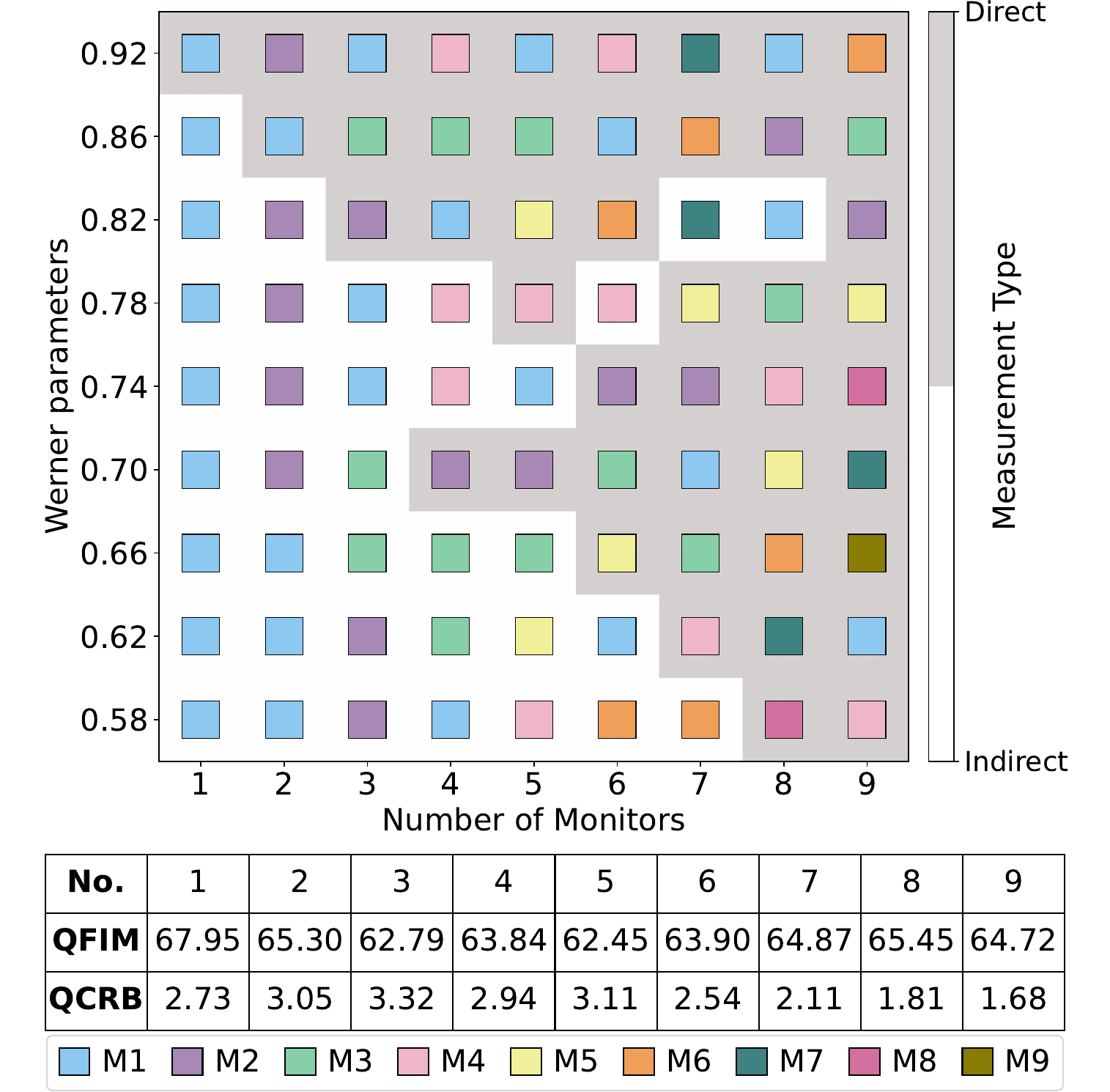}}%
\hfil
\subfloat[\label{fig2b}]{\includegraphics[scale=0.25]{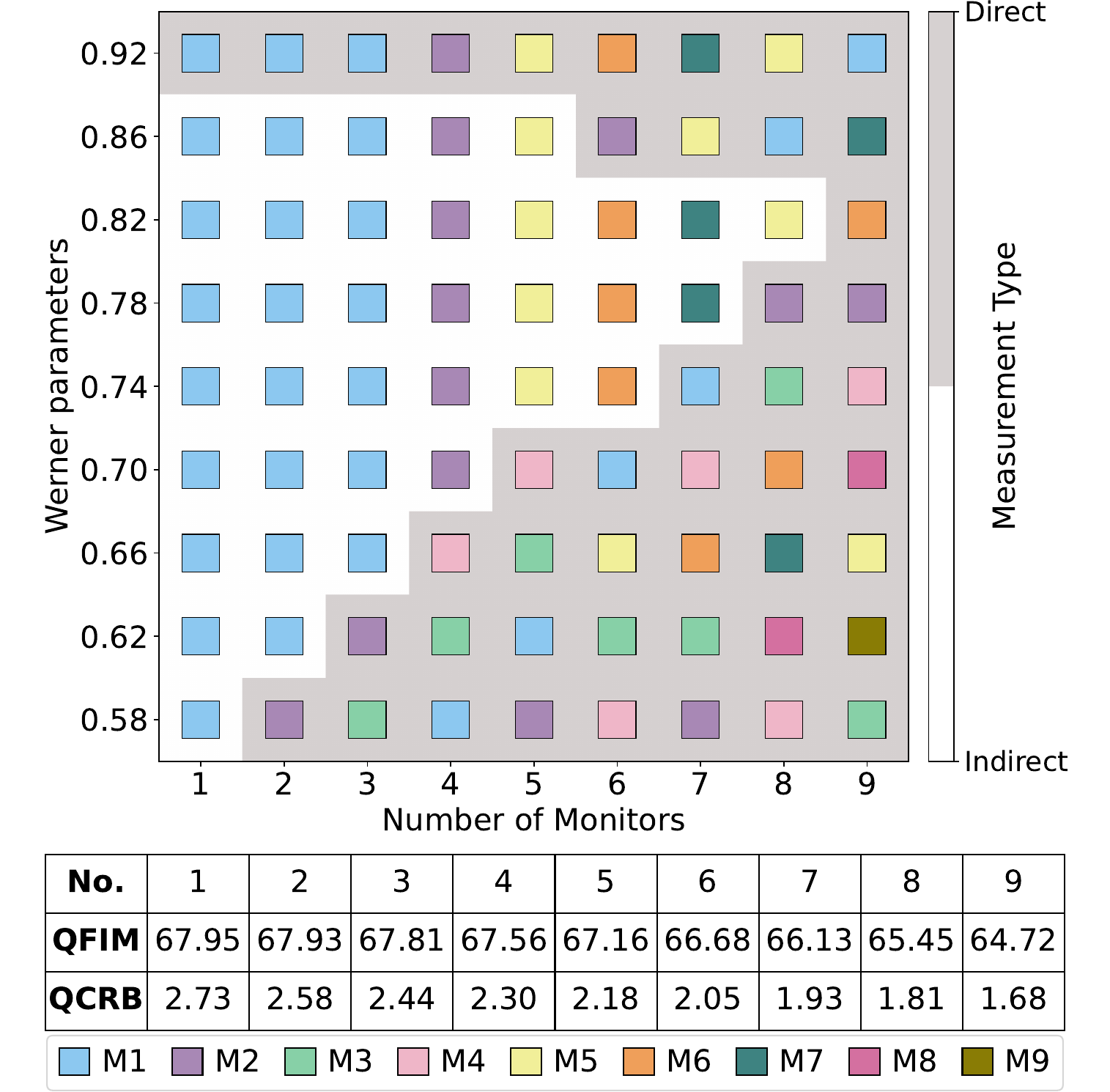}}%
\caption{Monitor placement results for a ten-node star network with inhomogeneous Werner parameters under the direct monitoring enforcement constraint~\eqref{9}, with up to nine available monitors. The x-axis denotes the number of deployed monitors, ranging from one to nine, while each row corresponds to a network link ordered in decreasing Werner parameter value. Indirect measurements are shown in white, and direct measurements are shown in light gray. The color inside each cell identifies the monitor assigned to the corresponding measurement. Subfigures show the results for (a) QMF and (b) QF. The corresponding QFIM trace and QCRB values for each number of deployed monitors are reported in the table below each subfigure. The variation in monitor index assignment is due to symmetry in the model, as all monitors are identical and interchangeable.}
\label{fig_sim_results1}
\end{figure*}
\begin{figure*}[t]
\centering
\subfloat[\label{fig2c}]{\includegraphics[scale=0.25]{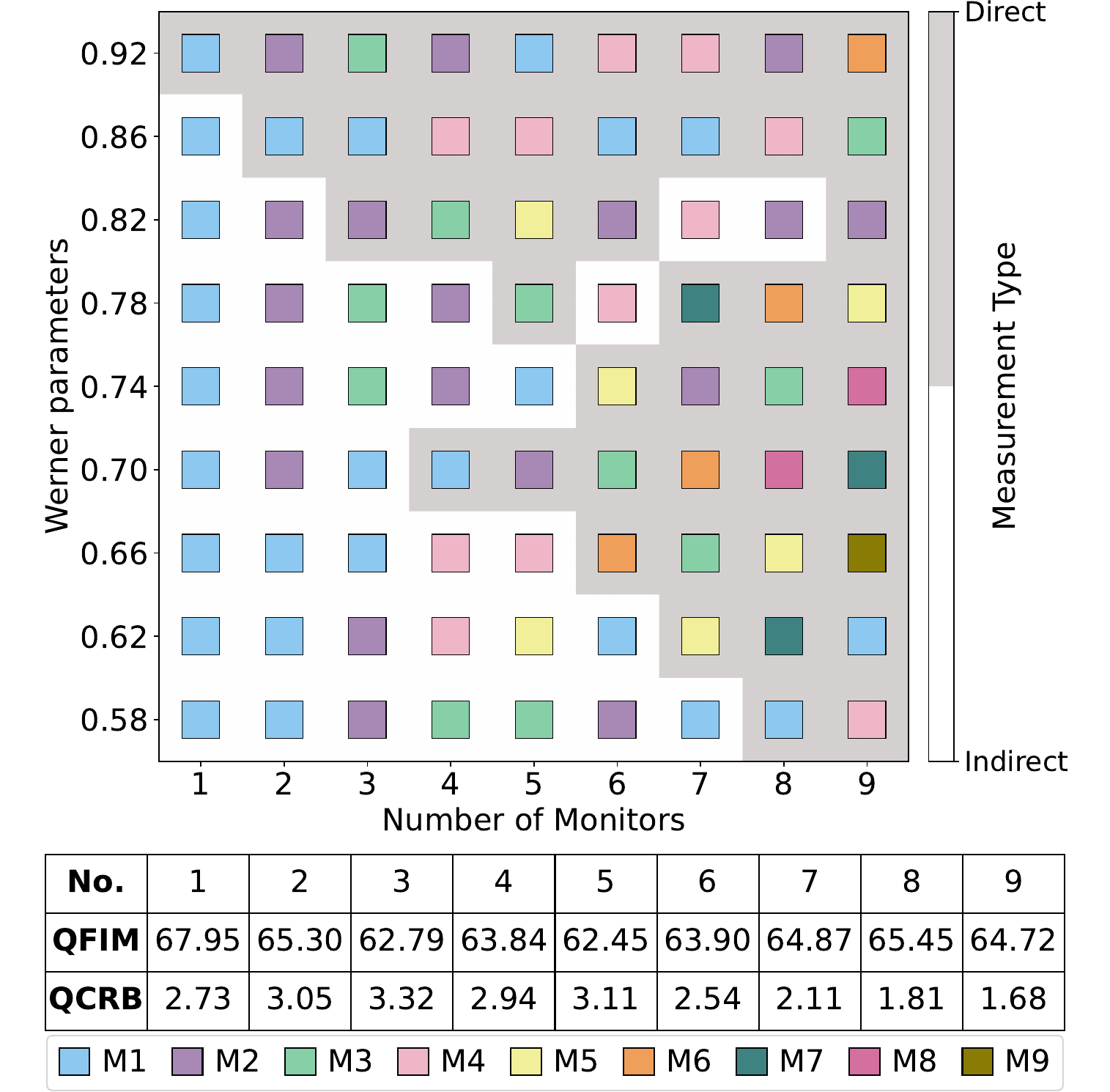}}%
\hfil
\subfloat[\label{fig2d}]{\includegraphics[scale=0.25]{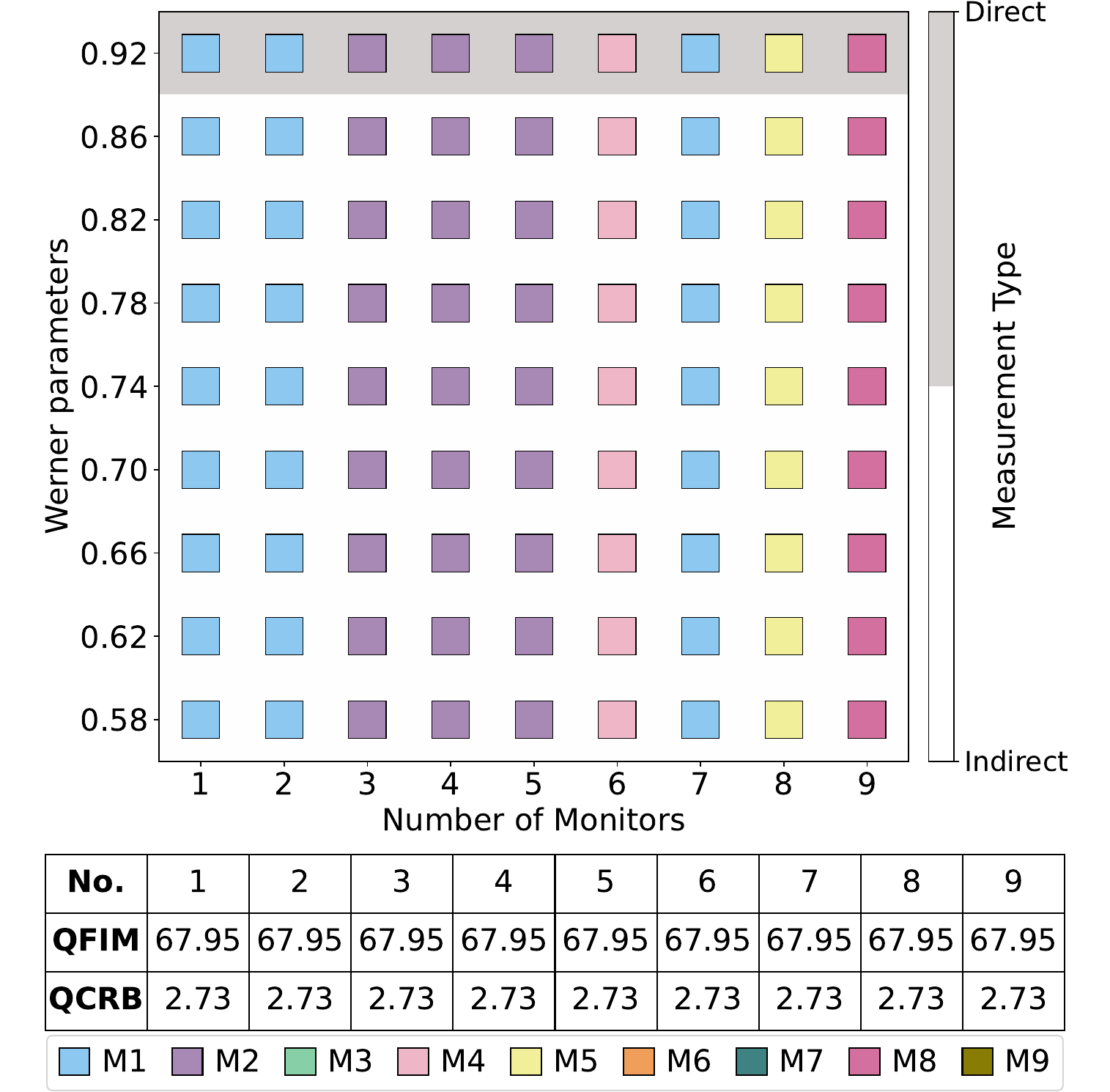}}%
\caption{Monitor placement results for a ten-node star network with inhomogeneous Werner parameters without direct monitoring enforcement constraint~\eqref{9}, with up to nine available monitors. The x-axis denotes the number of deployed monitors, ranging from one to nine, while each row corresponds to a network link ordered in decreasing Werner parameter value. Indirect measurements are shown in white, and direct measurements are shown in light gray. The color inside each cell identifies the monitor assigned to the corresponding measurement. Subfigures show the results for (a) QMF and (b) QF. The corresponding QFIM trace and QCRB values for each number of deployed monitors are reported in the table below each subfigure. The variation in monitor index assignment is due to symmetry in the model, as all monitors are identical and interchangeable.}
\label{fig_sim_results1a}
\end{figure*}

In stars, indirect monitoring is performed through two-hop paths passing through the central node regardless of network size. This probe homogeneity renders the formulation defined with constraints (\ref{lambh2}) and (\ref{lambh}) equivalent to the formulation with constraint (\ref{new}). We exclude the central node of the star from the candidate monitor locations, since placing a single monitor at the center would allow every link to be directly monitored, making the monitor placement problem trivial and eliminating the need to study indirect monitoring.

The monitor placement results for the ten-node star network with heterogeneous Werner parameters are presented in Fig.~\ref{fig_sim_results1} for the formulation with the direct monitoring enforcement constraint~\eqref{9}, and in Fig.~\ref{fig_sim_results1a} for the formulation without this constraint, with the number of monitors varied from one to nine. The QMF results are shown in Figs.~\ref{fig2a} and~\ref{fig2c}, while the QF results are shown in Figs.~\ref{fig2b} and~\ref{fig2d}. For both direct-enforced and non-enforced settings, the QMF formulation distributes indirectly monitored links more evenly among the available monitors, although more links are assigned to monitors incident to higher-quality links. For instance, four indirectly monitored links are assigned to the monitor incident to the link with the highest Werner parameter, while three indirectly monitored links are assigned to the monitor incident to the link with the second-highest Werner parameter. In contrast, in the QF formulation, all indirectly monitored links are assigned to the monitor incident to the link with the highest Werner parameter, which leads to monitor overloading. This concentration results from the T-optimal objective, since probing paths containing higher quality links provide larger QFIM contributions. Consequently, maximizing the QFIM trace favors assignments through the highest quality monitor rather than balancing the estimation information
across link parameters.

We now discuss the results obtained under the direct monitoring enforcement constraint. In the QF formulation as shown in Fig. \ref{fig2b}, the QCRB decreases as the number of monitors increases, from 2.73 for one monitor to 1.68 for nine monitors, indicating improved estimation precision with additional direct measurements. For QMF, shown in Fig.~\ref{fig2a}, the QCRB increases to 3.32 when three monitors are deployed and then decreases as more monitors become available. For eight and nine monitors, QMF and QF give identical QFIM and QCRB values, as the availability of monitors allows most links to be directly measured.

Without the direct-monitoring enforcement constraint, the QF formulation no longer benefits from increasing the number of deployed monitors, as shown in Fig.~4b. The QFIM and QCRB remain constant at $67.95$ and $2.73$, respectively, for all monitor counts. This behavior follows from
the T-optimal objective favoring the high-QFIM paths associated with the highest-quality monitor. Hence, additional monitors need not improve the poorly estimated parameter directions that dominate the QCRB. This differs from the direct-enforced QF case, where the QCRB decreases from 2.73 to 1.68 as more monitors are deployed. For QMF, the QFIM and QCRB values are unchanged after removing the direct monitoring enforcement constraint as shown in Fig. \ref{fig2c}. This indicates that for this formulation, the monitoring-overhead constraint already enforces sufficient distribution of measurement assignments across monitors. Therefore, direct monitoring enforcement has its strongest effect in the QF formulation.

We further evaluate both formulations by comparing their QCRBs, i.e., the trace of the QFIMs inverse, and monitoring-overhead values in both homogeneous and heterogeneous noise settings.

Fig. \ref{fig3a} shows the results for the homogeneous star network. Since all links have the same Werner parameter, the information contribution of each link is identical. Therefore, the QCRB depends mainly on the number of monitors rather than on their specific placement. As the number of monitors increases, more links can be directly monitored, and the trace of the QCRB decreases monotonically. The QF and QMF formulations give the same QCRB values both with and without Constraint (\ref{9}), showing that the direct-monitoring enforcement constraint does not change the estimation accuracy in the homogeneous case.

However, the monitoring overhead behaves differently for the two formulations. In the QF formulation, the objective focuses only on improving the QFIM and does not explicitly balance the assignments across monitors. As a result, the maximum monitoring overhead decreases gradually as additional monitors are introduced. In contrast, QMF explicitly accounts for monitoring overhead, leading to a more balanced distribution of link assignments. Thus, in the homogeneous case, QMF improves load balancing without sacrificing estimation accuracy.

Fig. \ref{fig3b} shows the corresponding results for the heterogeneous star network. In this case, the link qualities are different, so the monitor placement and link-assignment decisions strongly affect the QCRB. For QF with Constraint (\ref{9}), the QCRB decreases monotonically as the number of monitors increases. This is because the constraint forces links incident to monitors to be directly measured, so additional monitors increase the number of direct measurements and improve estimation accuracy. Without Constraint (\ref{9}), the QF formulation gives an almost constant QCRB and maintains a high monitoring overhead. This indicates that, when direct monitoring is not enforced, the optimizer can route many links through the monitor connected to high-quality link instead of using the newly placed monitors directly. It prevents the additional monitors from effectively reducing the overhead or improving the QCRB. For QMF, the overhead decreases significantly because the formulation penalizes high monitoring load. This load balancing can cause a temporary QCRB increase for intermediate monitor numbers, as some links are assigned to monitors associated with lower quality links. When the number of monitors increases, most links can be directly monitored, and the QF and QMF results converge.

Overall, these results show that Constraint (\ref{9}) is especially important in heterogeneous networks. It ensures that additional monitors are used for direct measurements whenever possible, which improves the QCRB and allows effective monitor placements. At the same time, QMF dominates QF under homogeneous noise conditions because it achieves the same precision as QF while attaining lower monitor overhead.

\begin{figure*}[t]
\centering
\subfloat[\label{fig3a}]{\includegraphics[scale=0.21]{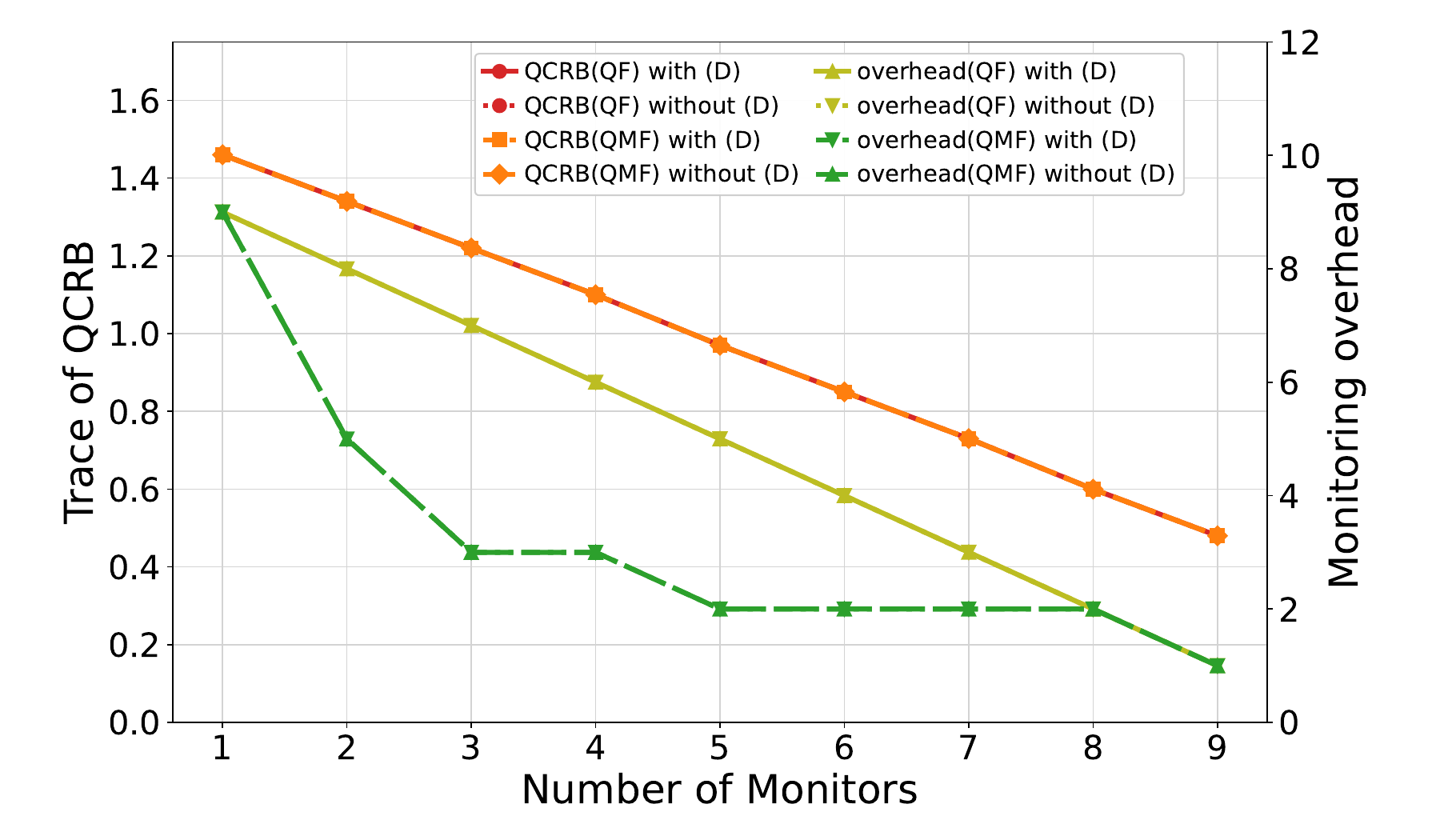}}%
\hfil
\subfloat[\label{fig3b}]{\includegraphics[scale=0.21]{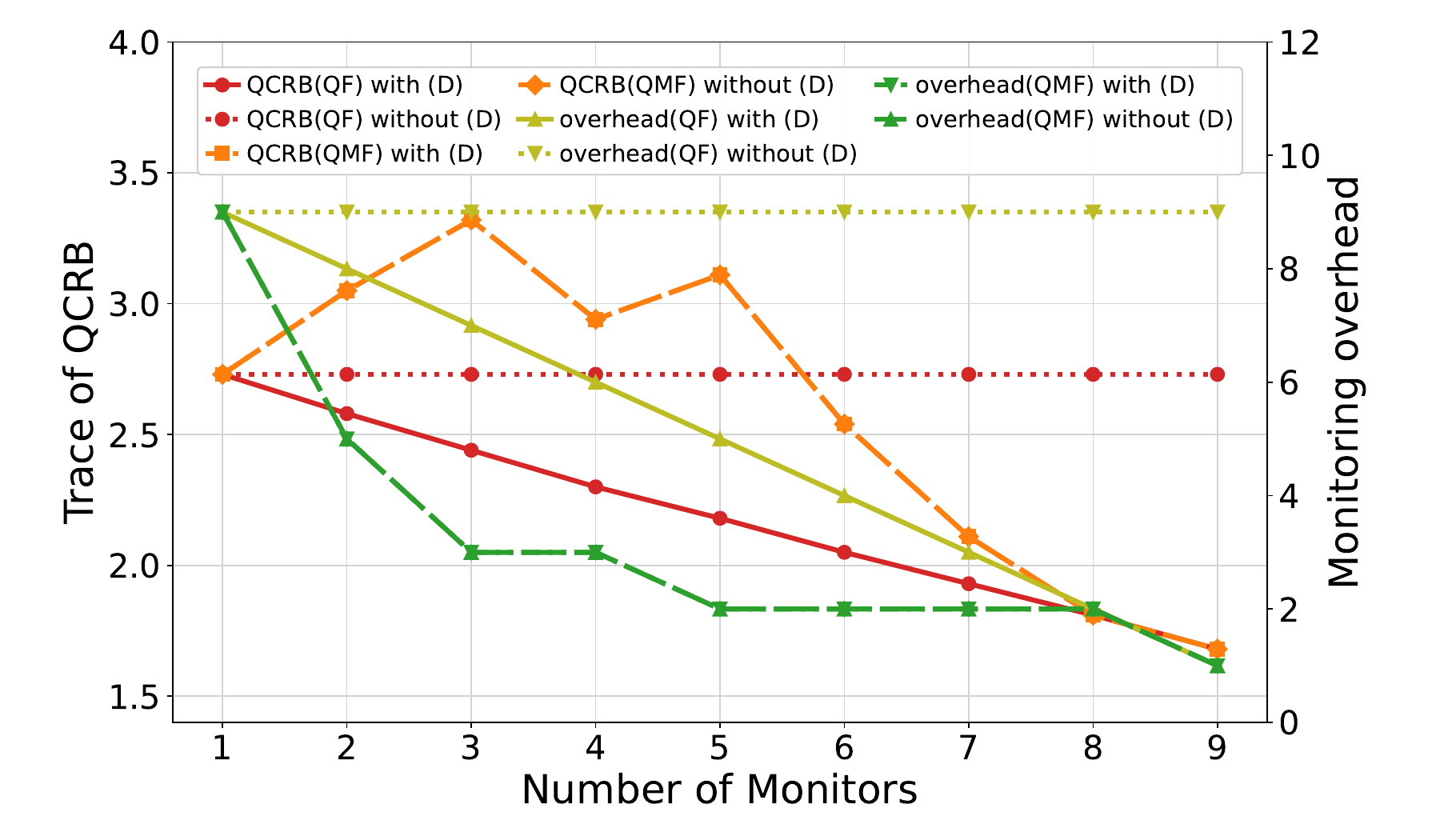}}%
\caption{QCRB and monitor overhead comparison for the ten-node star network with (a) uniform noise with all Werner parameters equal to 0.92 and (b) heterogeneous noise as shown in Figs.~\ref{fig_sim_results1},~\ref{fig_sim_results1a}. $D$ indicates direct monitoring enforcement constraint \eqref{9}.  }
\label{fig_sim_results2}
\end{figure*}

\subsubsection*{Tree Networks}
\begin{figure*}[t]
\centering
\subfloat[\label{fig4a}]{\includegraphics[scale=0.28]{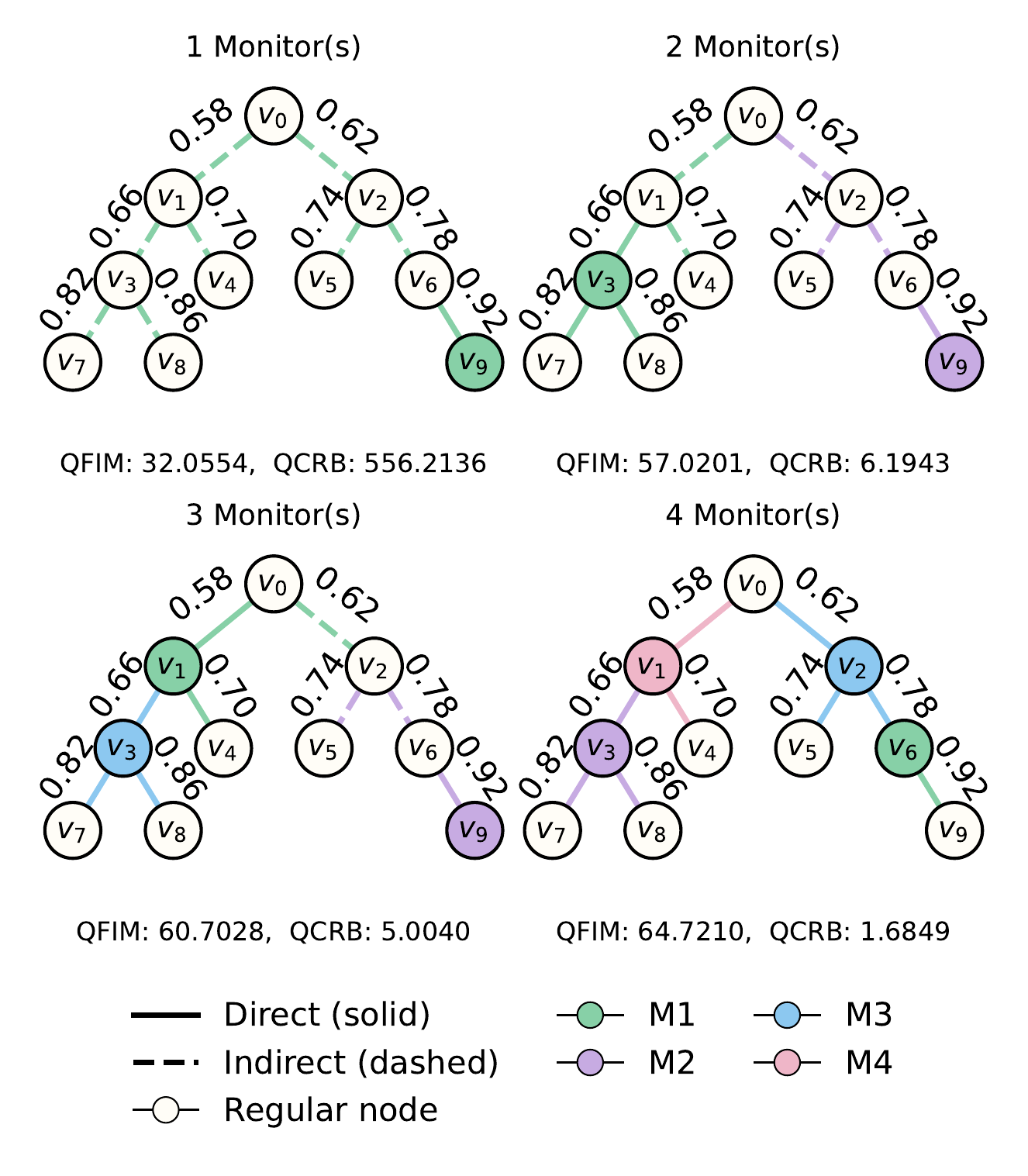}}%
\hfil
\subfloat[\label{fig4b}]{\includegraphics[scale=0.28]{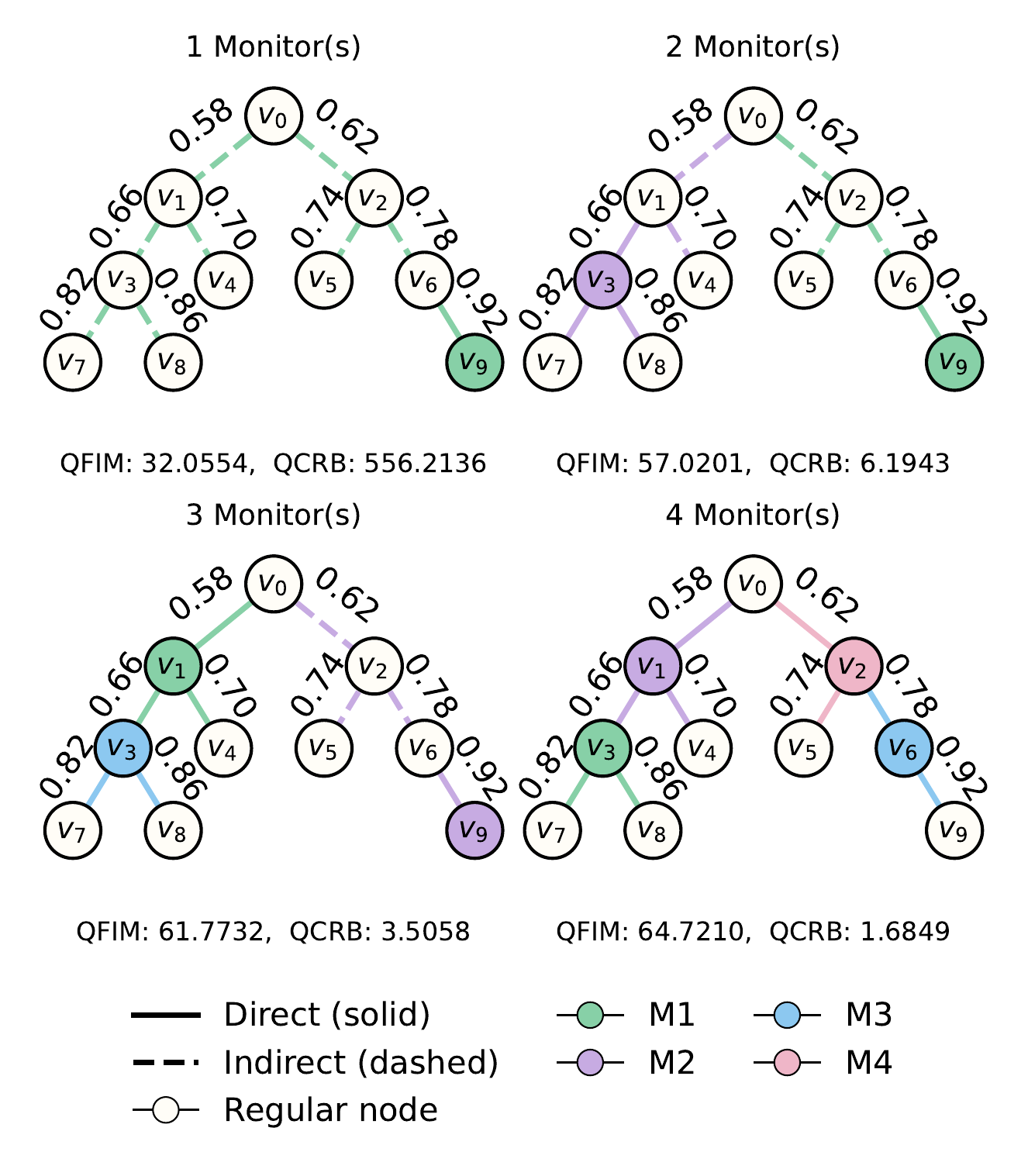}}%
\caption{Monitor placement results for the ten-node tree network with up to four monitors and inhomogeneous Werner parameters under direct monitoring enforcement constraint (\ref{9}). Node colors indicate monitor placements, and edge colors reflect the monitor performing the measurement. Solid edges denote links that are directly measured, while dashed edges represent indirect measurements. All links in the tree can be directly measured using four monitors, represented by four distinct colors, with each edge color corresponding to the monitor (of the same color) responsible for measuring that link. Subfigures show results from (a) QMF (b) QF. The variation in monitor index assignment is due to symmetry in the model, as all monitors are identical and interchangeable. }
\label{fig_sim_results4}
\end{figure*}

\begin{figure*}[t]
\centering
\subfloat[\label{fig4c}]{\includegraphics[scale=0.28]{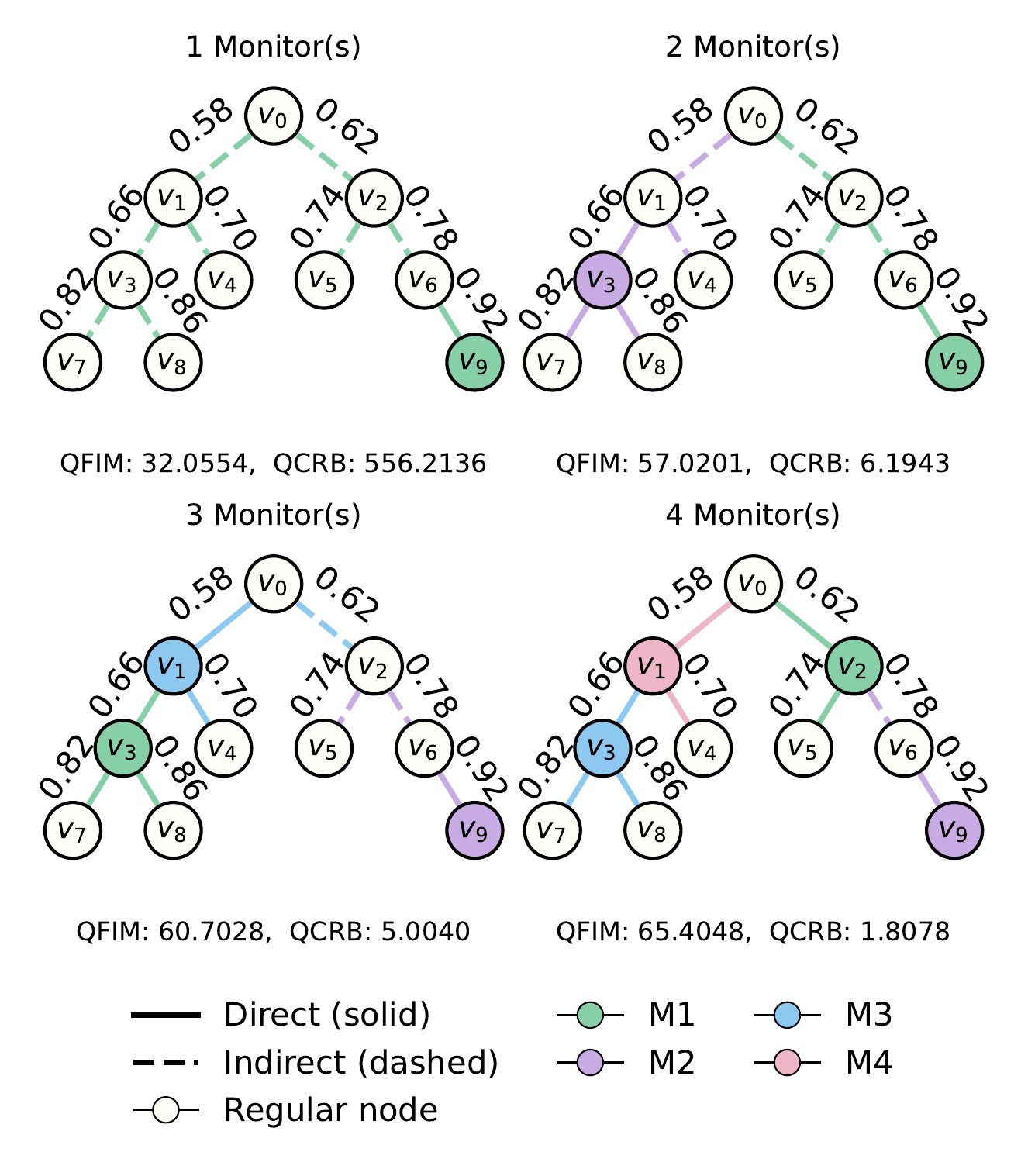}}%
                      \hfil
\subfloat[\label{fig4d}]{\includegraphics[scale=0.28]{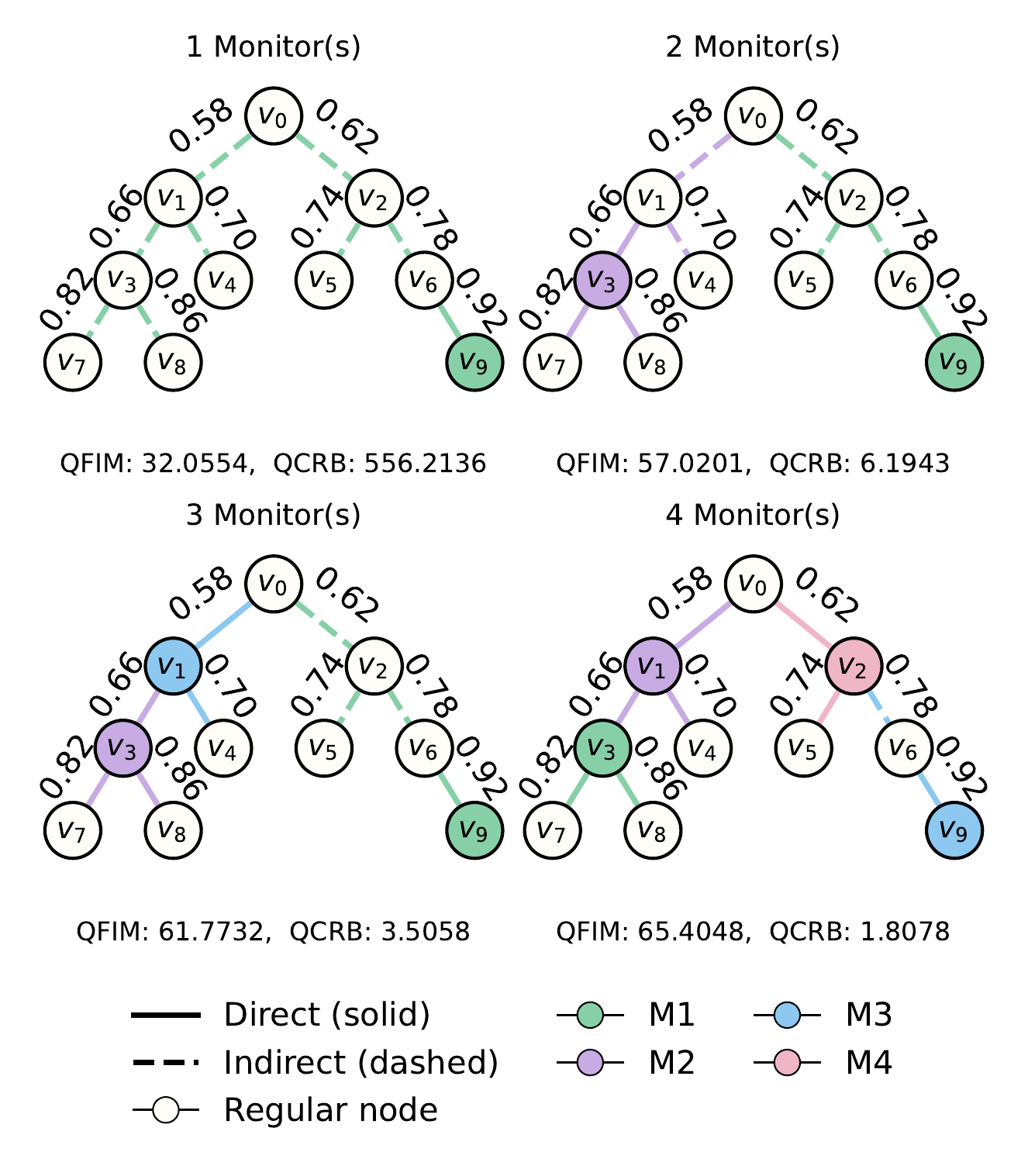}}%
\caption{Monitor placement results for the ten-node tree network with up to four monitors and inhomogeneous Werner parameters without direct monitoring enforcement constraint (\ref{9}). Node colors indicate monitor placements, and edge colors reflect the monitor performing the measurement. Solid edges denote links that are directly measured, while dashed edges represent indirect measurements. All links in the tree can be directly measured using four monitors, represented by four distinct colors, with each edge color corresponding to the monitor (of the same color) responsible for measuring that link. Subfigures show results from (a) QMF (b) QF. The variation in monitor index assignment is due to symmetry in the model, as all monitors are identical and interchangeable. }
\label{fig_sim_results4a}
\end{figure*}
 
Trees are more complex than stars and solving the optimal monitor placement on them allows probing the behavior of QMF and QF on less homogeneous topologies. Indirect probes in trees are, in general, paths with more than two hops. In this case, the ILP formulation defined with Constraints (\ref{lambh2}) and (\ref{lambh}) yields different solutions for the QF than the one defined with Constraint (\ref{new}). We focus on the formulation with (\ref{lambh2}) and (\ref{lambh}) and report results with (\ref{new}) in Appendix C. We also compare the results with and without direct monitoring enforcement constraint~\eqref{9}.

Similarly to the star case, for both direct-enforced and non-enforced cases as shown in Figs. \ref{fig_sim_results4} and \ref{fig_sim_results4a} respectively, QMF distributes the monitoring load more evenly across monitors while still prioritizing high-quality paths. 
For instance, with three monitors, covering the highest Werner parameters regions of the network while ensuring that no single monitor is overloaded. QF places monitors at the nodes connected to the links with the highest Werner parameters, and path assignments favor probes with high end-to-end Werner parameters. However, optimal monitor placement is dominated by path quality. Although QF achieves high estimation precision, it leads to more pronounced load concentration when compared to QMF. This behavior is exemplified in the case of three monitors, where the QF and QMF assignments yield a maximum monitoring load of five and three, respectively.

Under the direct monitoring enforcement constraint, QMF and QF give the same performance for one, two, and four monitors as shown in Figs. \ref{fig4a} and \ref{fig4b}. The large QCRB with one monitor is due to several links being measured indirectly through longer paths, while adding a second monitor improves estimation by increasing direct measurements and shortening indirect paths. The main difference occurs with three monitors, where QF achieves a lower QCRB than QMF, indicating that QMF trades some estimation accuracy for more balanced measurement assignments.

When the direct monitoring enforcement constraint is removed, the results remain unchanged for one, two, and three monitors as shown in Figs. \ref{fig4c} and \ref{fig4d}. However, with four monitors, both QMF and QF obtain a slightly higher QFIM but a worse QCRB compared with the direct-enforced case. This indicates that the direct monitoring enforcement constraint is beneficial as it improves QCRB performance.

Furthermore, the results for the tree network reveal that QMF explicitly encourages a more uniform distribution of monitoring tasks than QF, as expected. In contrast, QF may concentrate measurements on paths providing larger QFIM contributions, since its objective maximizes the total QFIM trace without explicitly accounting for monitoring overhead.

\section*{Discussion}\label{sec12}
This work demonstrates that the placement of monitor nodes, the allocation of quantum measurements and the design of probing protocols play a key role in the end-to-end characterization of quantum links in quantum network tomography. We proposed a cyclic sequential QNT protocol for arbitrary network topologies and developed a fixed-point learnability framework that establishes conditions under which all link parameters can be recovered from the selected probes. We further characterized global identifiability using the rank of the path-link incidence matrix and showed that the CSQP learnability conditions lead to a nonsingular QFIM, indicating that all link parameters are uniquely identifiable. For star topologies, we also derive closed form expressions for the maximum likelihood estimators and show that, for an interior feasible joint solution in the considered two-hop setting, the likelihood equation for a directly monitored link decouples from the indirect paths, reducing its MLE to the direct link estimate. By modeling the monitor placement and measurement assignment problem within an integer linear optimization framework, we show that different optimization objectives lead to distinct operational regimes: trace-driven formulations (QF) tend to concentrate measurements on nodes associated with high-quality links, while incorporating resource constraints (QMF) results in a more balanced distribution that supports scalable operation. The proposed framework is evaluated on star and hierarchical tree networks.

The present framework has several limitations. The optimization assumes prior knowledge of the link parameters and employs the QFIM trace as a surrogate objective to preserve the linearity of the optimization formulation. Furthermore, the present analysis is conducted under idealized assumptions; extending the framework to account for realistic non-idealities, including loss, decoherence, and imperfect Bell-state measurements, constitutes an important and challenging direction for future work.

\section*{Appendix A: Derivation of QFIM Expressions }\label{a1}
In this section, we present detailed derivations supporting the results discussed in the main text. Here, we compute the Quantum Fisher Information Matrix (QFIM) for the Werner-state parameters, which is applicable to arbitrary quantum networks. The derivation follows the mathematical framework outlined in \cite{paper}, with additional details provided here for completeness. Since the Werner states considered here are diagonal in the Bell basis and their eigenvectors are independent of the parameters, the QFIM can be expressed solely in terms of the eigenvalues $\lambda_k$ of the density matrix $\rho(\boldsymbol{w})$
\cite{qfim} as
\begin{equation}
\label{110}
F_{ig} = \sum_{k} \frac{1}{\lambda_k} \left( \frac{\partial \lambda_k}{\partial w_i} \right) \left( \frac{\partial \lambda_k}{\partial w_g} \right).
\end{equation}

\subsubsection*{Direct Monitoring}
The Werner state $\rho(w_i^2)$ corresponding to the direct measurement of a link $i$ based on the probe state distribution and measurement method described in the result section can be expressed as
\begin{equation}
\label{22}
\rho(w_i^2) = w_i^2 \, \phi_{00} + (1 - w_i^2)\frac{I_4}{4},
\end{equation}
where $\phi_{00}=\ket{\phi^+}\bra{\phi^+}$ denotes the Bell state and $I_4$ is the identity operator on the two-qubit Hilbert space. The eigenvalues are: 
\begin{align}
\lambda_1 = \frac{1 + 3w_i^2}{4}, 
\qquad
\lambda_{2,3,4} = \frac{1 - w_i^2}{4}.
\end{align}
Here, $\lambda_1$ corresponding to $\phi_{00}$ and $\lambda_{2,3,4}$ corresponding to other Bell states. Substituting these eigenvalues into Eq. (\ref{110}) and evaluating the derivatives with respect to $w_i$ yields a closed form expression for the QFIM corresponding to direct measurement.
\begin{align}
F_{ii}^{(D)} &= \frac{4}{1 + 3 w_i^2} 
\cdot \frac{\partial}{\partial w_i} \left( \frac{1 + 3 w_i^2}{4} \right) 
\cdot \frac{\partial}{\partial w_i} \left( \frac{1 + 3 w_i^2}{4} \right)\nonumber \\ 
&\quad + 3 \cdot \frac{4}{1 - w_i^2} 
\cdot \frac{\partial}{\partial w_i} \left( \frac{1 - w_i^2}{4} \right) 
\cdot \frac{\partial}{\partial w_i} \left( \frac{1 - w_i^2}{4} \right) \\[1em]
&= \frac{4}{1 + 3 w_i^2}. \left(\frac{6w_i}{4}\right)^2 + \frac{12}{1 - w_i^2}.\left(\frac{-2w_i}{4}\right)^2 \\[1em]
&= \frac{9 w_i^2}{1 + 3 w_i^2} + \frac{3 w_i^2}{1 - w_i^2} \\[1em]
F_{ii}^{(D)} &= \frac{12 w_i^2}{(1 + 3 w_i^2)(1 - w_i^2)} \label{1}
\end{align}

\subsubsection*{Indirect Monitoring}
For an indirect monitoring path $P_{ab}$, let
\begin{align}
    \rho_{ab} = \rho\!\left( \prod_{l \in P_{ab}} w_l^{\,2} \right)
\end{align}
denote the end-to-end Werner state generated by entanglement distribution and swapping along the path $P_{ab}$ between nodes $a$ and $b$. The eigenvalues are therefore
\begin{align}
\lambda_1 = \frac{1+3\prod_{l \in P_{ab}} w_l^{\,2}}{4}, 
\qquad
\lambda_{2,3,4} = \frac{1-\prod_{l \in P_{ab}} w_l^{\,2}}{4}.
\end{align}

Substituting the eigenvalues into Eq.~(\ref{110}) and computing the derivatives with respect to $w_i$ and $w_g$ gives the closed form expression for the QFIM corresponding to indirect measurement.

\begin{align}
F_{ig}^{(I)} &= \frac{4}{1 + 3 \prod_{l \in P_{ab}} w_l^2} \cdot 
\frac{\partial}{\partial w_i} \left( \frac{1 + 3 \prod_{l \in P_{ab}} w_l^2}{4} \right)
\cdot
\frac{\partial}{\partial w_g} \left( \frac{1 + 3 \prod_{l \in P_{ab}} w_l^2}{4} \right)\nonumber \\[1em]
&\quad + 3 \cdot \frac{4}{1 - \prod_{l \in P_{ab}} w_l^2} \cdot 
\frac{\partial}{\partial w_i} \left( \frac{1 - \prod_{l \in P_{ab}} w_l^2}{4} \right)
\cdot
\frac{\partial}{\partial w_g} \left( \frac{1 - \prod_{l \in P_{ab}} w_l^2}{4} \right) \\[1em]
&= \frac{4}{1 + 3 \prod_{l \in P_{ab}} w_l^2} \cdot 
\left( \frac{9}{4} w_i w_g \prod_{\substack{{l \in P_{ab}}\\l \ne i}} w_l^2 
\prod_{\substack{{l \in P_{ab}}\\l \ne g}} w_l^2 \right) \nonumber \\[0.5em]
&\quad + \frac{12}{1 - \prod_{l \in P_{ab}} w_l^2} \cdot 
\left( \frac{1}{4} w_i w_g \prod_{\substack{{l \in P_{ab}}\\l \ne i}} w_l^2 
\prod_{\substack{{l \in P_{ab}}\\l \ne g}} w_l^2 \right) \\[1em]
&= 
\frac{9 w_i w_g}{1 + 3 \prod_{l \in P_{ab}} w_l^2} 
\cdot \prod_{\substack{{l \in P_{ab}}\\l \ne i}} w_l^2 \cdot \prod_{\substack{{l \in P_{ab}}\\l \ne g}} w_l^2 
+ 
\frac{3 w_i w_g}{1 - \prod_{l \in P_{ab}}w_l^2} 
\cdot \prod_{\substack{{l \in P_{ab}}\\l \ne i}} w_l^2 \cdot \prod_{\substack{{l \in P_{ab}}\\l \ne g}} w_l^2\\[1em]
F_{ig}^{(I)} &= \frac{12 w_i w_g}{(1 + 3 \prod_{l \in P_{ab}} w_l^2)(1 - \prod_{l \in P_{ab}}w_l^2)} \bigg( \prod_{\substack{{l \in P_{ab}}\\l \ne i}} w_l^2 \cdot \prod_{\substack{{l \in P_{ab}}\\l \ne g}} w_l^2\bigg). \label{2}
\end{align}

\subsection*{Numerical Analysis using QFIM}
\begin{figure*}[t]
\centering
\subfloat[\label{figml}]{\includegraphics[scale=0.25]{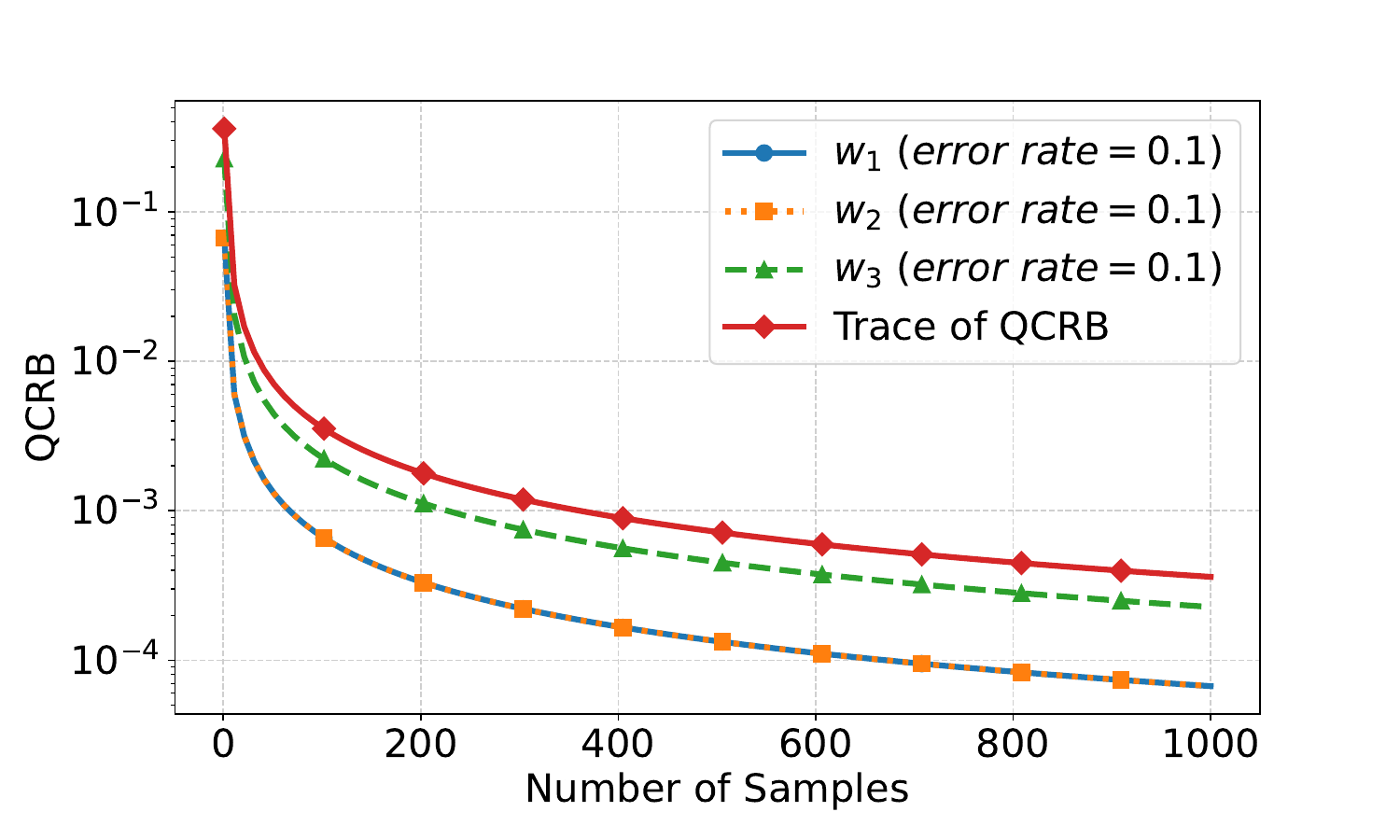}}%
\hfil
\subfloat[\label{figcrb}]{\includegraphics[scale=0.25]{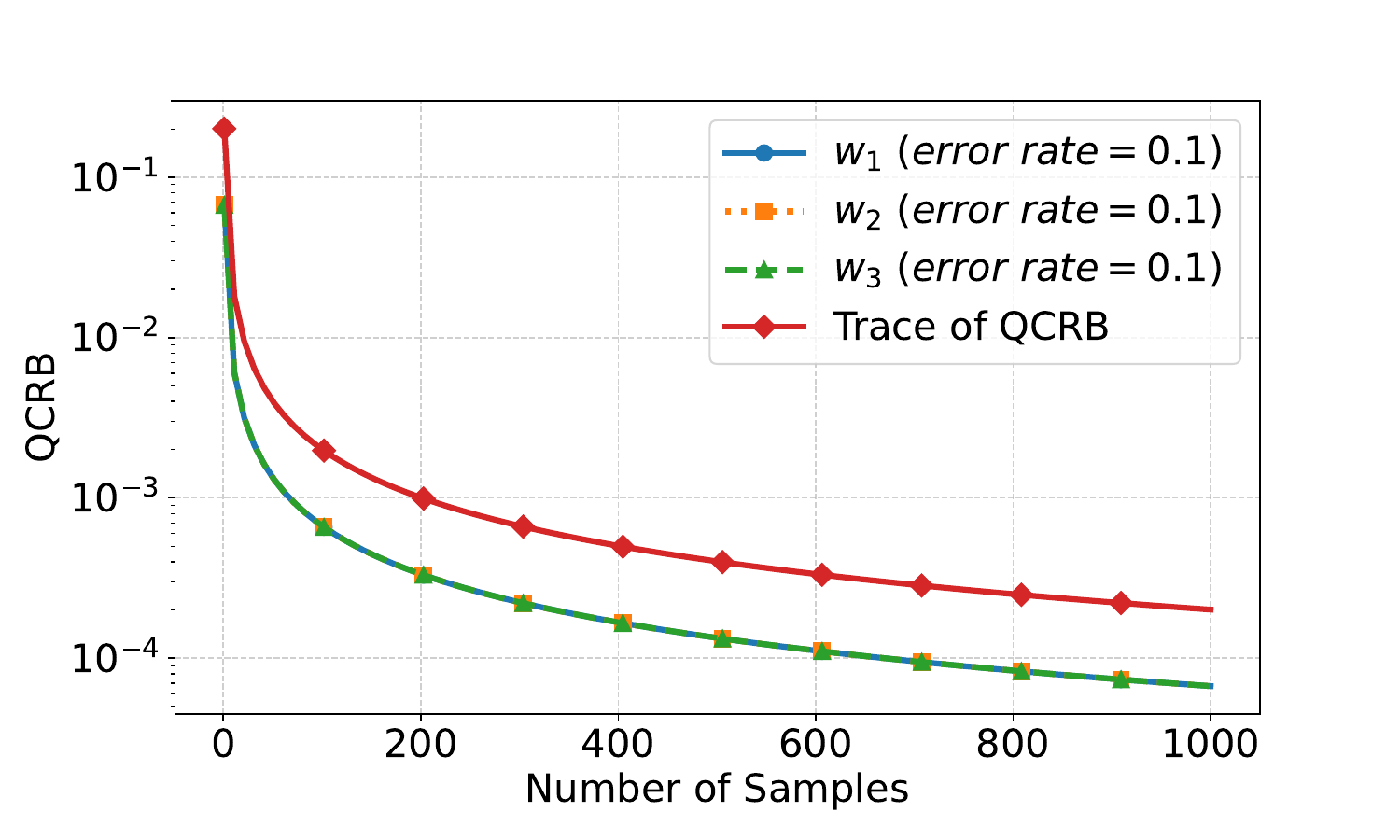}}%
\caption{Numerical Analysis using QFIM: (a) QCRB of the three estimators and QCRB trace, with error rate = 0.1 in two-monitor configuration, (b) QCRB of three estimators and trace of QCRB in Three-Monitor configuration for error rate = 0.1.}
\end{figure*}
We now compare two- and three-monitor configurations in a four-node star network, where $v_0$ denotes the hub node and $v_1$, $v_2$, and $v_3$ denote the end nodes. Fig. \ref{figml} presents the QCRB values of the estimators $w_1$, $w_2$, and $w_3$ corresponding to the links $e_1=(v_0,v_1)$, $e_2=(v_0,v_2)$, and $e_3=(v_0,v_3)$, under the two-monitor configuration with monitors placed at the nodes $v_1$ and $v_2$. The results show that the QCRB values for $w_1$ and $w_2$ are identical, while the QCRB value for $w_3$ is significantly higher. This disparity reflects a lower estimation precision for $w_3$, which is indirectly monitored. Fig. \ref{figcrb} displays the QCRB values for the same three estimators in a three-monitor configuration with monitors placed at the nodes $v_1$, $v_2$, and $v_3$. In this setup, the QCRB values for $w_1$, $w_2$, and $w_3$ are equal, demonstrating uniform and improved estimation accuracy across all links when directly measured by the monitor.

\section*{Appendix B: Derivation of MLE Expressions }\label{a2}
In this section, we provide the derivation of the maximum likelihood estimators (MLEs) used for Werner-parameter estimation in the $n$-node star network. The MLE $\hat{w}$ is the parameter value that maximizes the likelihood of the observed measurement outcomes. For a parameter $w$, the MLE is defined as
\begin{equation}
\label{eq:supp_mle_definition}
    \hat{w}
    =
    \underset{w}{\arg\max}\; \ell(D|w),
\end{equation}
where $D$ denotes the measurement dataset and $\ell(D|w)$ denotes the likelihood of observing $D$ under parameter $w$.

For a path $P_{ab}$, let $N^{ab}$ denote the total number of measurements performed on the probe associated with $P_{ab}$, and let $N_{00}^{ab}$ denote the number of outcomes equal to $\phi_{00}$.

\subsection*{Direct Monitoring}
Consider a directly monitored link $e_i=(v_a,v_b)$, so that the corresponding probe path is $P_{ab}=e_i$. The cyclic probe has effective Werner parameter $w_i^2$, and the resulting state is
\begin{equation}
\label{eq:supp_direct_state}
    \rho(w_i^2)
    =
    w_i^2 \phi_{00}
    +
    (1-w_i^2)\frac{I_4}{4}.
\end{equation}
The Bell-basis outcome probabilities are
\begin{align}
    \Pr(\phi_{00})
    &=
    \frac{1+3w_i^2}{4},\\
   \Pr(\phi_{mn})
&=
\frac{1-w_i^2}{4},
\qquad mn\in\{01,10,11\}.
\end{align}

The log-likelihood is given by
\begin{align}
\label{eq:supp_direct_loglikelihood}
    \log \ell(D|w_i)
    =
    N_{00}^{ab}
    \log\left(\frac{1+3w_i^2}{4}\right)
    +
    \left(N^{ab}-N_{00}^{ab}\right)
    \log\left(\frac{1-w_i^2}{4}\right).
\end{align}
Differentiating with respect to $w_i$ gives
\begin{align}
    \frac{\partial \log \ell}{\partial w_i}
    &=
    \frac{6N_{00}^{ab}w_i}{1+3w_i^2}
    -
    \frac{2\left(N^{ab}-N_{00}^{ab}\right)w_i}{1-w_i^2}.
\end{align}
Setting the derivative to zero and multiplying by
$4(1+3w_i^2)(1-w_i^2)$ gives
\begin{align}
    24N_{00}^{ab}w_i(1-w_i^2)
    -
    8\left(N^{ab}-N_{00}^{ab}\right)w_i(1+3w_i^2)
    =
    0.
\end{align}
This reduces to
\begin{equation}
    32N_{00}^{ab}
    -
    8N^{ab}
    -
    24N^{ab}w_i^2
    =
    0.
\end{equation}
Hence, the MLE of the directly monitored link parameter is
\begin{equation}
\label{eq:supp_direct_mle}
    \hat{w}_i
    =
    \sqrt{
    \frac{4N_{00}^{ab}-N^{ab}}{3N^{ab}}
    }.
\end{equation}

\subsection*{Indirect Monitoring}
We now consider the indirect monitoring setting for the star network. Let $v_a\in M$ be the monitor node, let $v_h$ denote the center node of the star, and let
\begin{equation}
e_q=(v_a,v_h)
\end{equation}
be the directly monitored common link. Let $\mathcal{I}_q$ denote the set of indices of links indirectly monitored through paths that contain $e_q$. For each $j\in\mathcal I_q$, the indirectly monitored link is
\begin{equation}
e_j=(v_h,v_j),
\end{equation}
and the corresponding indirect path is
\begin{equation}
P_{aj}=[v_a,v_h,v_j].
\end{equation}
Its end-to-end Werner parameter is
\begin{equation}
w_{aj}=w_qw_j.
\end{equation}
Thus, the corresponding probe state is
\begin{equation}
\label{eq:supp_indirect_state}
    \rho(w_{aj}^2)
    = w_{aj}^2 \phi_{00} + (1 - w_{aj}^2) \frac{I_4}{4}
    = w_q^2 w_j^2 \phi_{00} + (1 - w_q^2 w_j^2) \frac{I_4}{4}.
\end{equation}
The Bell-basis outcome probabilities are
\begin{align}
    \Pr(\phi_{00}) &= \frac{1 + 3w_q^2 w_j^2}{4}, \\
    \Pr(\phi_{mn})
&=
\frac{1-w_q^2w_j^2}{4},
\qquad mn\in\{01,10,11\}.
\end{align}

For each indirect path $P_{aj}$, define
\begin{equation}
    \label{eq:xaj}
    x_{aj} = w_{aj}^2 = w_q^2 w_j^2.
\end{equation}
The log-likelihood associated with the measurements on $P_{aj}$ is
\begin{equation}
\label{eq:supp_indirect_loglikelihood}
    \log \ell(D \mid x_{aj})
    = N_{00}^{aj} \log\left(\frac{1 + 3x_{aj}}{4}\right)
    + \left(N^{aj} - N_{00}^{aj}\right) \log\left(\frac{1 - x_{aj}}{4}\right).
\end{equation}
Differentiating with respect to $x_{aj}$ and setting the derivative to zero
yields
\begin{equation}
\label{71}
    \frac{3N_{00}^{aj}}{1 + 3x_{aj}} - \frac{N^{aj} - N_{00}^{aj}}{1 - x_{aj}} = 0.
\end{equation}
Solving for $x_{aj}$ yields the path-wise Maximum Likelihood Estimate (MLE):
\begin{equation}
\label{eq:supp_xaj_mle}
    \hat{x}_{aj} = \widehat{w_{aj}^2} = \frac{4N_{00}^{aj} - N^{aj}}{3N^{aj}}.
\end{equation}
Equivalently, the MLE of the end-to-end Werner parameter on path $P_{aj}$ is
\begin{equation}
\label{eq:supp_waj_mle}
    \hat{w}_{aj} = \sqrt{\frac{4N_{00}^{aj} - N^{aj}}{3N^{aj}}}.
\end{equation}

Since $x_{aj}=w_q^2w_j^2$ and $\partial x_{aj}/\partial w_j = 2w_q^2w_j \neq 0$ for $w_q,w_j\in(0,1)$, the likelihood condition with respect to $w_j$ is equivalent to the likelihood condition with respect to $x_{aj}$ in Eq.~\eqref{71}. For the interior joint MLE considered here, the path-wise estimates are assumed to satisfy the physical feasibility condition $0<\hat{x}_{aj}<\hat{w}_q^2$ for all $j \in \mathcal{I}_q$.

The direct link $e_q = (v_a, v_h)$ appears in both the direct probe $P_{ah} = e_q$ and the indirect probes $\{P_{aj} : j \in \mathcal{I}_q\}$. Therefore, estimating $w_q$ incorporates contributions from both direct and indirect measurements. The joint log-likelihood terms involving $w_q$ are
\begin{align}
\label{eq:supp_common_joint_likelihood}
    \log \ell_q
    &= N_{00}^{ah} \log\left(\frac{1 + 3w_q^2}{4}\right)
    + \left(N^{ah} - N_{00}^{ah}\right) \log\left(\frac{1 - w_q^2}{4}\right) \nonumber \\
    &\quad + \sum_{j \in \mathcal{I}_q} \left[ N_{00}^{aj} \log\left(\frac{1 + 3w_q^2 w_j^2}{4}\right) + \left(N^{aj} - N_{00}^{aj}\right) \log\left(\frac{1 - w_q^2 w_j^2}{4}\right) \right].
\end{align}
Differentiating with respect to $w_q$ gives
\begin{align}
\label{eq:supp_common_score}
    \frac{\partial \log \ell_q}{\partial w_q}
    &= \frac{6 N_{00}^{ah} w_q}{1 + 3w_q^2} - \frac{2\left(N^{ah} - N_{00}^{ah}\right) w_q}{1 - w_q^2} \nonumber \\
    &\quad + \sum_{j \in \mathcal{I}_q} \left[ \frac{6 N_{00}^{aj} w_q w_j^2}{1 + 3w_q^2 w_j^2} - \frac{2\left(N^{aj} - N_{00}^{aj}\right) w_q w_j^2}{1 - w_q^2 w_j^2} \right].
\end{align}

Multiplying Eq. \eqref{eq:supp_common_score} by $w_q$ and using $w_q w_j^2 = \frac{\hat{x}_{aj}}{w_q}$ from Eq. \eqref{eq:xaj}, we define the indirect contribution term $C_q$ of Eq. \eqref{eq:supp_common_score} as
\begin{equation}
\label{eq:supp_Cq_def}
    C_q = \sum_{j \in \mathcal{I}_q} \left[ \frac{6 N_{00}^{aj} \hat{x}_{aj}}{1 + 3\hat{x}_{aj}} - \frac{2\left(N^{aj} - N_{00}^{aj}\right) \hat{x}_{aj}}{1 - \hat{x}_{aj}} \right].
\end{equation}
\\
Each summand in Eq. \eqref{eq:supp_Cq_def} is $2\hat{x}_{aj}$ times the left-hand side of the likelihood equation in Eq. \eqref{71}, evaluated at $\hat{x}_{aj}$. At an interior feasible joint MLE satisfying the condition above, Eq. \eqref{71} holds for every $j \in \mathcal{I}_q$. Therefore, every indirect contribution vanishes and $C_q=0$.

The likelihood equation for $w_q$ therefore reduces to the contribution from the direct probe alone.
\begin{equation}
\label{eq:supp_common_score_compact}
    \left( \frac{6 N_{00}^{ah}}{1 + 3w_q^2} - \frac{2\left(N^{ah} - N_{00}^{ah}\right)}{1 - w_q^2} \right) w_q^2 =0.
\end{equation}
Solving for $w_q$ gives
\begin{equation}
\label{eq:neweqmle2}
    \hat{w}_q
    = \sqrt{\frac{4N_{00}^{ah} - N^{ah}}{3N^{ah}}}.
\end{equation}
After obtaining the common-link estimate $\hat{w}_q$, each indirectly monitored link $e_j$ ($j \in \mathcal{I}_q$) is estimated by isolating $w_j$ in \eqref{eq:xaj} using the end-to-end MLE \eqref{eq:supp_waj_mle}:
\begin{equation}
\label{eq:supp_indirect_mle}
    \hat{w}_j
    = \frac{\hat{w}_{aj}}{\hat{w}_q}
    = \sqrt{\frac{4N_{00}^{aj} - N^{aj}}{3N^{aj} \hat{w}_q^2}},
    \qquad j \in \mathcal{I}_q.
\end{equation}
If the feasibility condition is violated, the constrained joint MLE may lie on the parameter boundary, and the decoupling need not hold.

\subsection*{Estimation Performance Analysis}
\begin{figure*}[t]
\centering
\subfloat[\label{fig3}]{\includegraphics[scale=0.25]{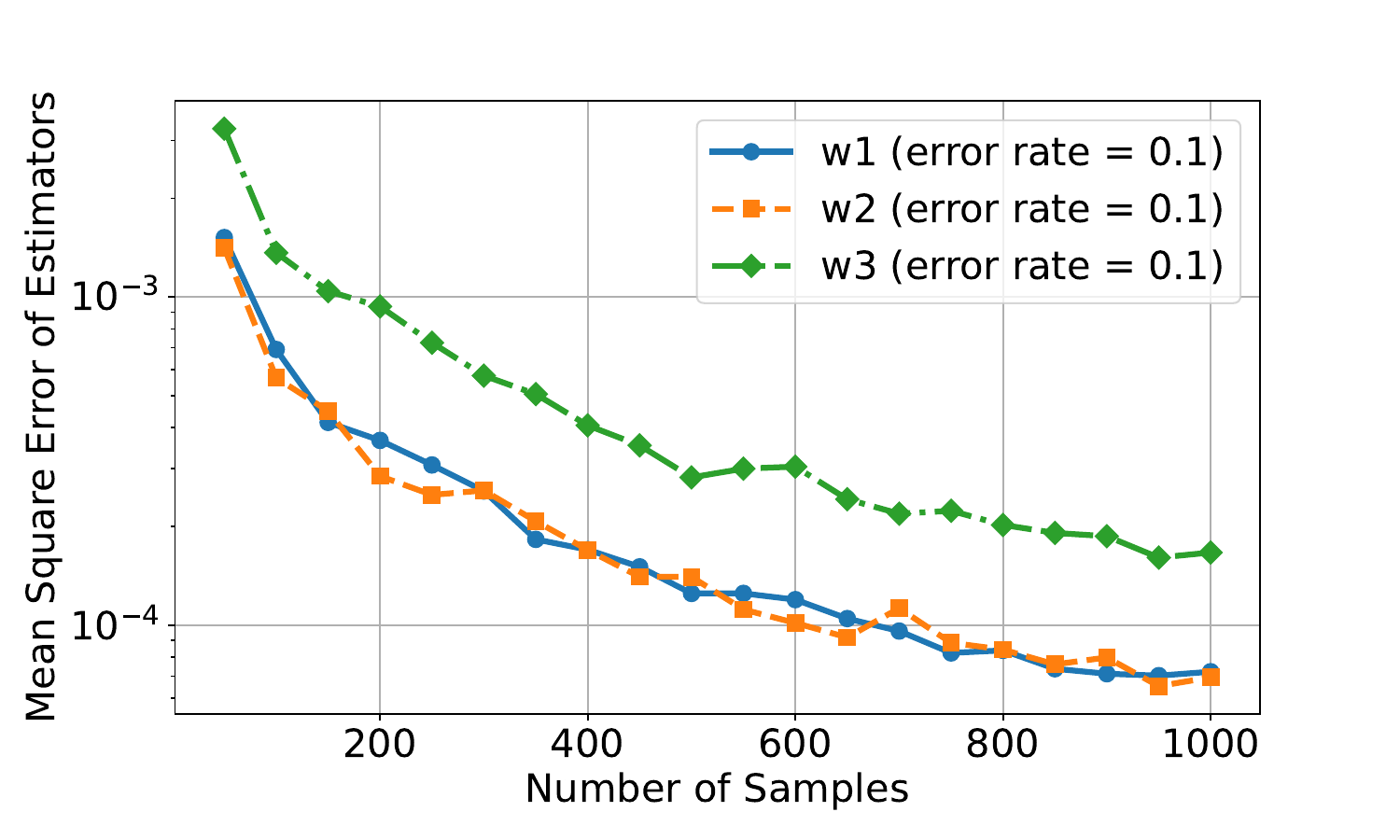}}%
\hfil
\subfloat[\label{fig4}]{\includegraphics[scale=0.25]{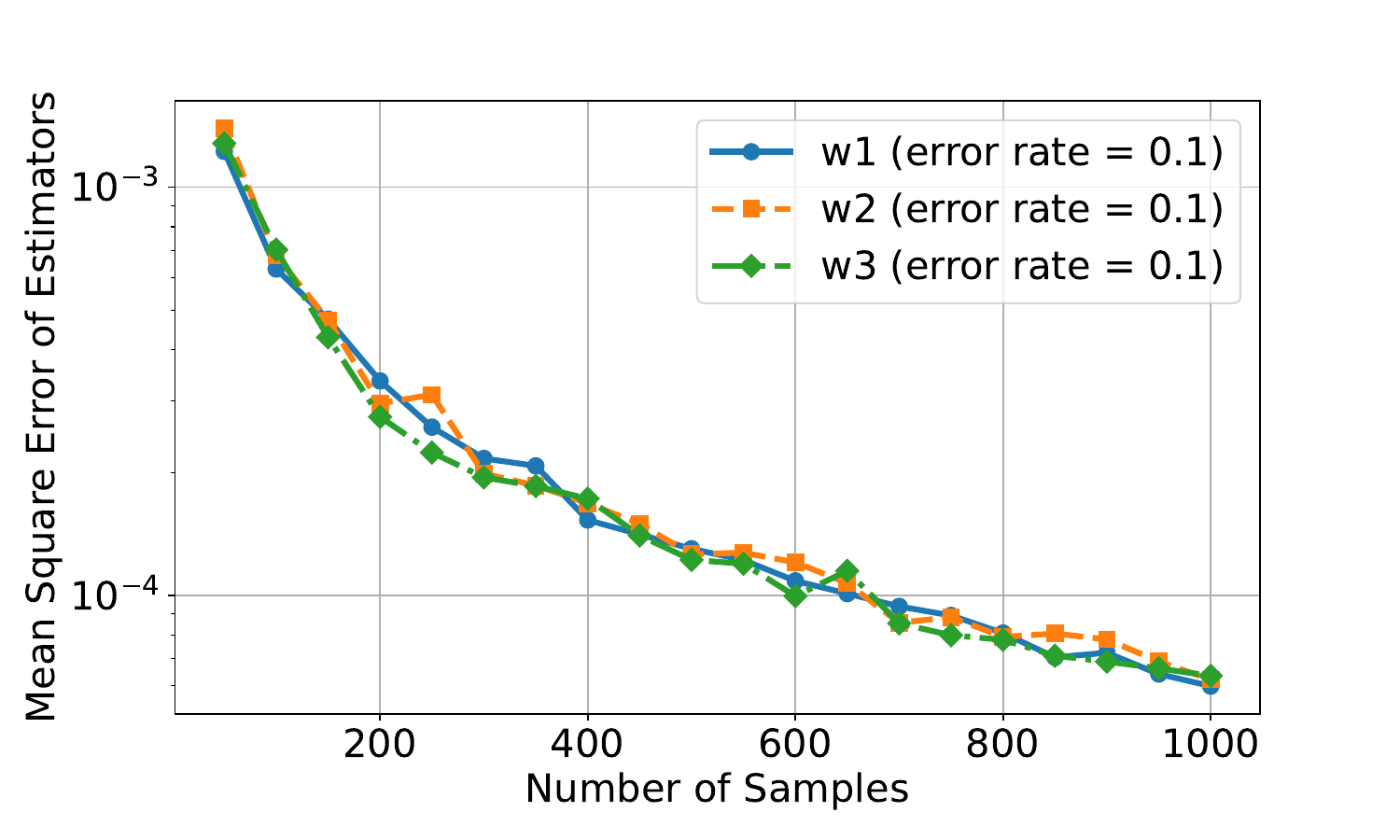}}%
\caption{Maximum Likelihood Estimator for Werner parameter: (a) Mean square error of MLE for Two Monitor Setup with error rate=0.1 for all three estimators, (b) Mean square error of MLE for Three-Monitor Setup with error rate=0.1 for all three estimators.}
\label{fig_sim_results3}
\end{figure*}
Using the derived MLE estimators and the QCRB as a benchmark, we next analyze the finite-sample performance of the maximum likelihood estimators and their convergence toward the QCRB for the four-node star network.
Figs. \ref{fig3} and \ref{fig4} present the mean square error (MSE) of the MLEs as a function of the number of samples for network configurations with two and three monitors. The values of the Werner parameters $w_1$, $w_2$, and $w_3$ are fixed at 0.9. The results indicate that the MSE decreases as the number of samples increases, demonstrating improved estimator accuracy with larger data sets. In the three-monitor setup, all estimators exhibit similar behavior and converge to the same value, suggesting uniform estimation accuracy across the network. In contrast, in the two-monitor setup, the estimator corresponding to the link without a directly connected monitor converges to a higher MSE compared to the other two estimators, which are associated with links that have direct monitor connections. This observation aligns with our theoretical expectations, confirming that the presence of a monitor directly impacts estimation accuracy. Furthermore, the observed MSE trends closely resemble those in the QCRB plots shown in Figs. \ref{figml} and \ref{figcrb}, indicating that the estimators effectively approach the lower bound set by the QCRB.

\section*{Appendix C: Additional Results } \label{a4}
\begin{figure*}[t]
\centering
\subfloat[\label{fig44a}]{\includegraphics[scale=0.28]{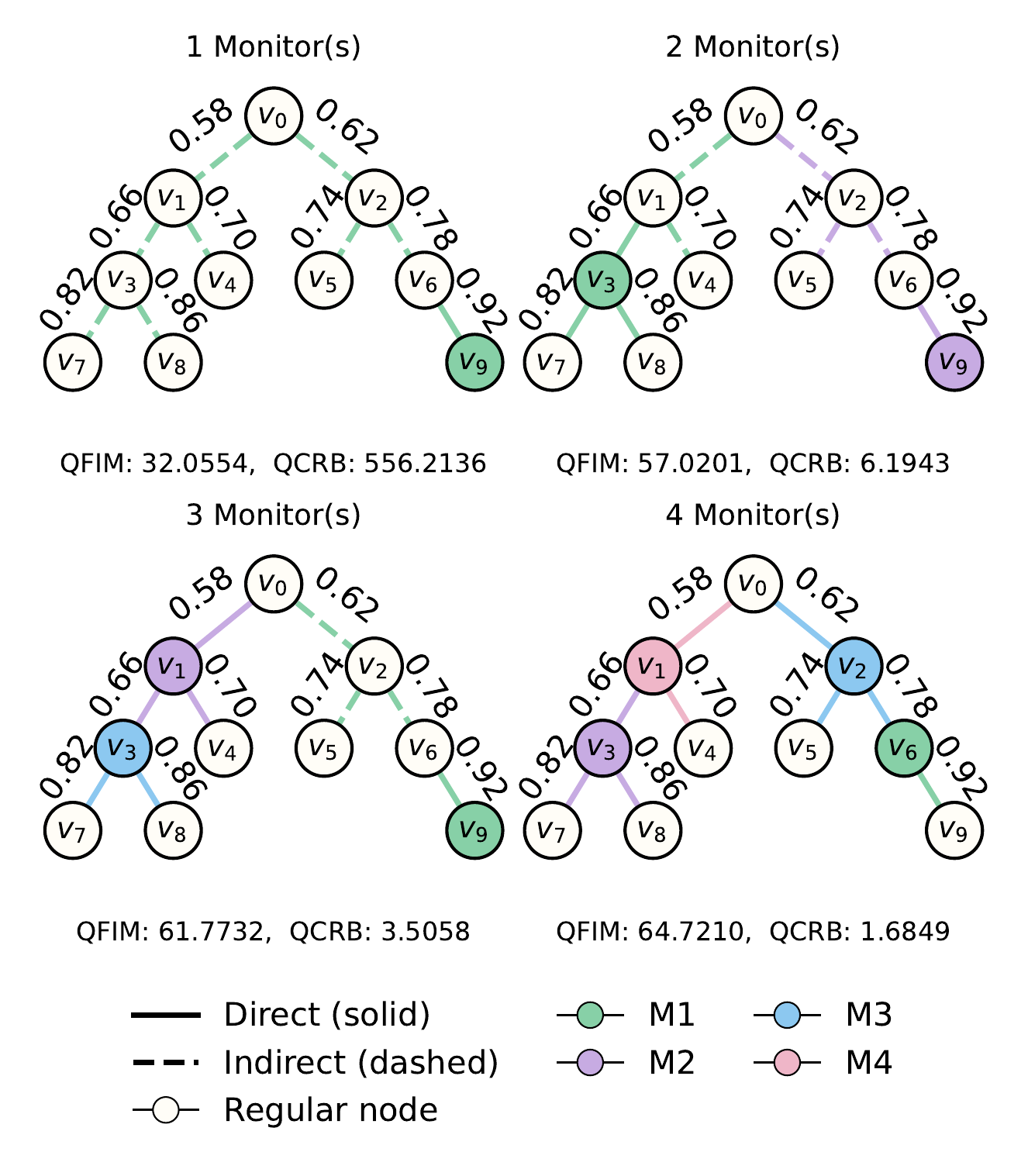}}%
\hfil
\subfloat[\label{fig44b}]{\includegraphics[scale=0.28]{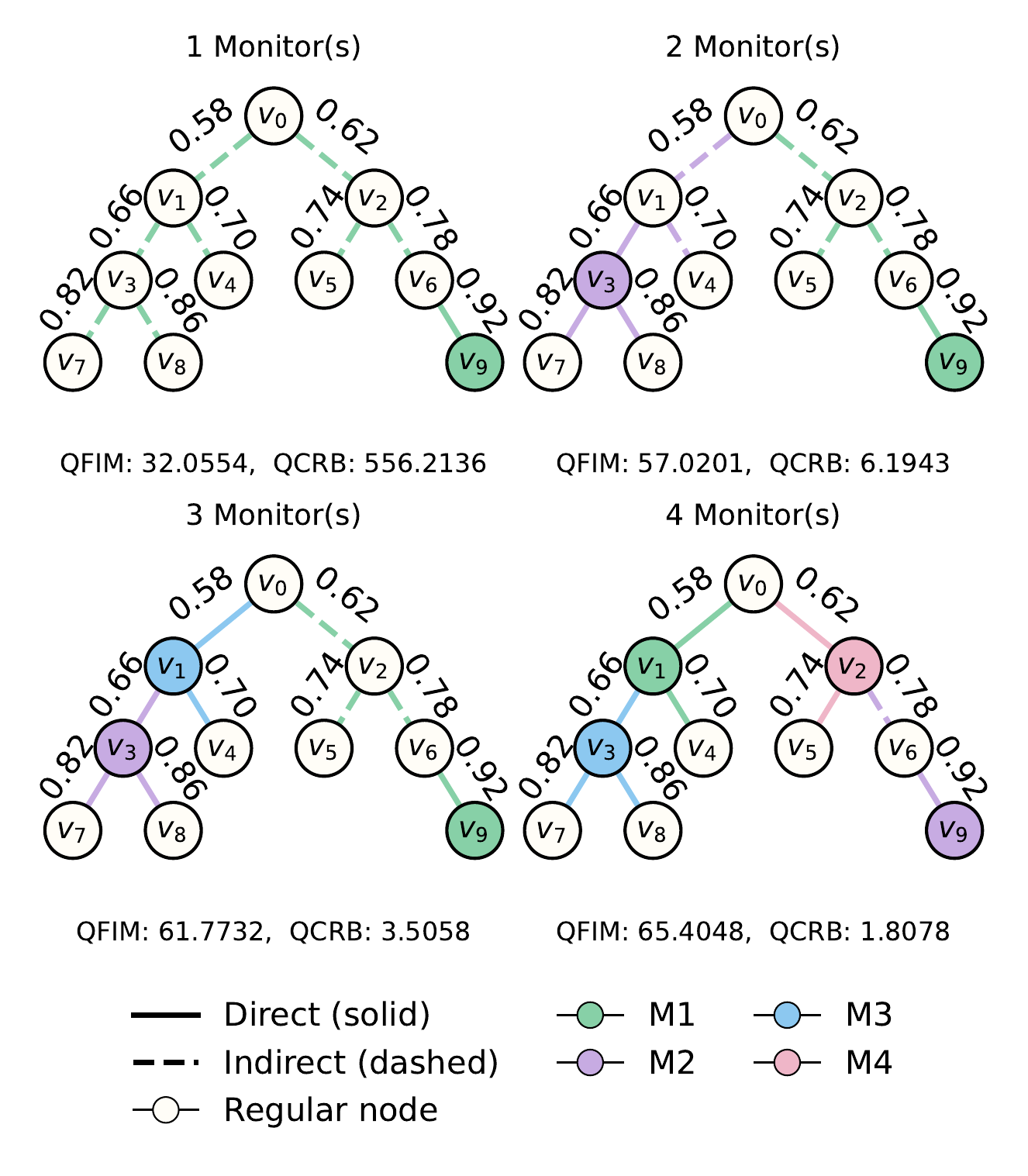}}%
\caption{Monitor placement results for the ten-node tree network with up to four monitors using Constraint (\ref{new}) instead of Constraints (\ref{lambh}) and (\ref{lambh2}) from the main text. Node colors indicate monitor placements, and edge colors reflect the monitor performing the measurement. Solid edges denote links that are directly measured, while dashed edges represent indirect measurements. All links in the tree can be directly measured using four monitors, represented by four distinct colors, with each edge color corresponding to the monitor (of the same color) responsible for measuring that link. Subfigures show results from (a) QF with constraint (\ref{9}) (b) QF without constraint (\ref{9}).}
\label{fig_sim_results333}
\end{figure*}
Now, we present the QF results obtained by replacing Constraints (\ref{lambh}) and (\ref{lambh2}) with Constraint (\ref{new}), as described in the main text. The QMF results are not shown because the link assignments remain unchanged with and without Constraint (\ref{new}). Fig. \ref{fig44a} shows the QF results with direct-monitoring enforcement. Compared with the formulation using Constraints (\ref{lambh}) and (\ref{lambh2}), the only assignment change occurs for four monitors, where node $v_6$ is assigned one link instead of two (Fig. \ref{fig4b}). However, the QFIM and QCRB values remain unchanged. Similarly, Fig. \ref{fig44b} shows the results without Constraint (\ref{9}); in this case, node $v_3$ is assigned three links instead of two (Fig. \ref{fig4d}), again with no change in the QFIM or QCRB. Overall, replacing Constraints (\ref{lambh}) and (\ref{lambh2}) with Constraint (\ref{new}) does not affect the estimation performance of the ten-node tree network, although some link assignments change.

\bmhead{Acknowledgements}
This research was supported by Research Ireland grant 21/US-C2C/3750 for CoQREATE (Convergent Quantum REsearch Alliance in Telecommunications), CONNECT-2 grant 13/RC/2077\_P2, and by the NSF-ERC Center for Quantum Networks grant EEC- 1941583.

\section*{Declarations}

\subsection*{Data Availability}
All data generated and analyzed during the study are available from the corresponding author upon reasonable request.

\subsection*{Code Availability}
The codes used for generating and analyzing data are available from the
corresponding author upon reasonable request.

\subsection*{Author contributions}
A.K.R. and M.G.A. contributed equally to this work. N.M., D.T., D.K., and I.D. provided supervision and guidance throughout this project. All authors contributed to the discussions and to the refinement of the manuscript.

\subsection*{Competing interests}
The authors declare no competing interests.

\bibliography{sn-bibliography}

\end{document}